\documentclass[11pt]{article}
\usepackage[english]{babel}
\usepackage{fullpage}
\usepackage[utf8]{inputenc}
\usepackage{amsmath,amsthm,amsfonts,mathtools}
\usepackage{xcolor}
\usepackage[hidelinks]{hyperref}
\usepackage{cleveref}
\usepackage{tcolorbox}
\usepackage{wrapfig}

\usepackage[linesnumbered,ruled,vlined]{algorithm2e}

\title{Non-obvious Manipulability in Hedonic Games with Friends Appreciation Preferences}
\author{Michele Flammini$^{1,2}$, Maria Fomenko$^1$, Giovanna Varricchio$^2$}
\date{%
	$^1$Gran Sasso Science Institute, L'Aquila, Italy\\%
	$^2$University of Calabria, Rende, Italy\\%
		michele.flammini@gssi.it, maria.fomenko@gssi.it, giovanna.varricchio@unical.it
	\\[2ex]%
	\today
}

\usepackage{amsthm}
\usepackage{enumitem}
\usepackage{cleveref}
\usepackage{marginnote}
\usepackage{todonotes}

\usepackage{tikz}
\usetikzlibrary{calc} 
\usetikzlibrary{arrows}

\usepackage{subcaption}

\usepackage{microtype}
\usepackage{varwidth}

\usepackage{thm-restate}

\newtheorem{proposition}{Proposition}
\newtheorem{observation}{Observation}
\newtheorem{theorem}{Theorem}
\newtheorem{lemma}{Lemma}

\theoremstyle{definition}
\newtheorem{example}{Example}
\newtheorem{definition}{Definition}

\newcommand{\agents}{\mathcal{N}}

\newcommand{\SW}{\mathsf{SW}}
\newcommand{\set}[1]{\{#1\}}
\newcommand{\ceil}[1]{\left\lceil#1\right\rceil}
\newcommand{\floor}[1]{\left\lfloor#1\right\rfloor}
\newcommand{\modulus}[1]{\left\lvert #1 \right\rvert }
\newcommand{\FA}{\mathsf{FA}}
\newcommand{\EA}{\mathsf{EA}}
\newcommand{\outcomes}{\Pi}
\newcommand{\GC}{\mathcal{GC}}
\newcommand{\Gf}{G^f}
\newcommand{\threePartition}{$3$-\textsc{Partition}}
\newcommand{\iIsFriend}[1]{F_{#1}^{-1}}

\newcommand{\goct}{gOct}
\newcommand{\neigh}{N}
\newcommand{\improveSW}{\textsc{ImproveSW}}
\newcommand{\randMech}{\textsc{RandMech}}
\newcommand{\edg}{edg}

\newcommand{\OPT}{\textsf{OPT}}
\newcommand{\opt}{\textsf{opt}}
\newcommand{\instance}{\mathcal{I}}

\newcommand{\mech}{\mathcal{M}}
\newcommand{\declared}{\mathbf{d}}
\newcommand{\declarations}{\mathcal{D}}

\newcounter{mechanismCounter}
\newtheorem{mecha}{Mechanism}
\newenvironment{mechanism}
{\mecha \addtocounter{mechanismCounter}{1}}
{\endmecha}

\crefname{mechanism}{Mechanism}{Mechanism}

\DeclareMathOperator{\argmax}{arg\,max}

\usepackage[ruled,vlined]{algorithm2e}
\SetKwInOut{Input}{Input}\SetKwInOut{Output}{Output}

\begin{document}
\maketitle

\begin{abstract}
In this paper, we study non-obvious manipulability (NOM), a relaxed form of strategyproofness, in the context of Hedonic Games (HGs) with Friends Appreciation ($\FA$) preferences. In HGs, the aim is to partition agents into coalitions according to their preferences which solely depend on the coalition they are assigned to. Under $\FA$ preferences, agents consider any other agent either a friend or an enemy, preferring coalitions with more friends and, in case of ties, the ones with fewer enemies. Our goal is to design mechanisms that prevent manipulations while optimizing social welfare.

Prior research established that computing a welfare maximizing (optimum) partition for $\FA$ preferences is not strategyproof, and the best-known approximation to the optimum subject to strategyproofness is linear in the number of agents. In this work, we explore NOM to improve approximation results. We first prove the existence of a NOM mechanism that always outputs the optimum; however, we also demonstrate that the computation of an optimal partition is NP-hard. To address this complexity, we focus on approximation mechanisms and propose a NOM mechanism guaranteeing a $(4+o(1))$-approximation in polynomial time.

Finally, we briefly discuss NOM in the case of Enemies Aversion ($\EA$) preferences, the counterpart of $\FA$, where agents give priority to coalitions with fewer enemies and show that no mechanism computing the optimum can be NOM.
\end{abstract}

\section{Introduction}
Hedonic Games (HGs) \cite{dreze1980hedonic} offer a game-theoretic framework for understanding the coalition formation of selfish agents and have been extensively studied in the literature (e.g.,\cite{aziz2013pareto,aziz2013computing,banerjee2001core,Bogomolnaia02,elkind2009hedonic,elkind2020price,gairing2010computing}).
In such games, the objective is to partition a set of agents into disjoint coalitions, with each agent's satisfaction determined solely by the members of her coalition. 
Depending on the nature of the agents' preferences, several HGs classes arise that may capture various social interactions between agents. 
For example, in \emph{additively separable} HGs (ASHGs)~\cite{hajdukova2004coalition}, agents evaluate coalitions by summing up the values they assign to every other participant;
in HGs with \emph{friends appreciation} ($\FA$) preferences~\cite{dimitrov2004enemies}, each agent splits the others into friends and enemies and prefers coalitions with more friends, in case of ties, she favors the ones with less enemies.

While most of the existing literature on HGs has put attention on the existence and computation of several stability concepts based on either individual~\cite{Bloch2011,Feldman15,gairing2010computing} or group deviations~\cite{Bogomolnaia02,banerjee2001core,elkind2009hedonic,gairing2010computing,igarashi2016hedonic}, a recent stream of research is focusing on
designing \emph{strategyproof} (SP) mechanisms. Such mechanisms prevent agents from manipulating the outcome by misrepresenting their preferences while ensuring desirable properties like stability or a reasonable approximation to the maximum \emph{social welfare} -- the sum of the agents' utilities in the outcome.
Unfortunately, achieving strategyproofness with good social welfare guarantees is challenging. For ASHGs no bounded approximation mechanism is SP~\cite{flammini2021strategyproof}, and even in the simple case of $\FA$ preferences the best-known SP mechanism guarantees an approximation linear in the number of agents~\cite{flammini2022strategyproof}. 

More broadly, strategyproofness has been widely studied in several game-theoretic settings, nonetheless, such a requirement is often incompatible with other desirable properties or even impossible to achieve \cite{ozyurt2009general, amanatidis2017truthful, brandl2018proving}. Moreover, there exist mechanisms that are not strategyproof in a strict sense, but, in order to successfully manipulate, an agent has to possess the knowledge of others' strategies and deeply understand the underlying mechanics. Otherwise, she might end up with an outcome that is even worse than the one she attempted to avoid. However, the ability of a cognitively limited agent to satisfy this requirement seems unrealistic, which has led to the notion of \emph{non-obvious manipulability} (NOM) introduced to distinguish the mechanisms that can be easily manipulated from the ones that are unlikely to be manipulated in practice~\cite{troyan2020obvious}.

\paragraph{Our Contribution.} With this paper, we initiate the study of NOM in the context of HGs focusing on $\FA$ preferences. We aim at improving upon the performances, in terms of the social welfare guarantee, of SP mechanisms in this setting. To this end, we begin by analyzing the structure of optimal outcomes and give a deeper understanding of how such outcomes look like for some interesting instances. We specifically focus on some structures of friendship relationships (that we will call generalized octopus graphs) which turn out to be very useful for providing a picture of socially optimal outcomes. This enables us to show that there always exists a NOM mechanism computing a social welfare maximizing partition (\Cref{thm:optIsNOM}). Unfortunately, we also show that finding such an outcome is NP-hard (\Cref{thm:OPThard}).
We, therefore, propose a $(4+o(1))$-approximation mechanism that is also NOM and works in polynomial time (\Cref{thm:apxIsNOM}). This mechanism provides a significant improvement over the existing strategyproof mechanism, and, besides NOM, it constitutes the first constant approximation for the problem of maximizing the utilitarian for $\FA$ preferences. 
Finally, we investigate NOM for {\em enemies aversion} ($\EA$) preferences, the natural counterpart of $\FA$ where agents give priority to coalitions with fewer enemies, and show that no optimal mechanism is NOM (\Cref{thm:EAnotNOM}). This shows that NOM, albeit it might be considered a weak notion, is not always compatible with optimality.


\paragraph{Related Work.}
HGs with $\FA$ and $\EA$ preferences have been widely studied and further extended to capture more involved social contexts~\cite{dimitrov2006simple,rothe2018borda,kerkmann2022altruistic,barrot2019unknown}. In this stream of research, a systematic analysis of stable outcomes has been provided~\cite{chen2023Hedonic}. In addition to stability, strategyproofness has also been considered: In~\cite{dimitrov2004enemies} the authors show that for $\FA$ and $\EA$ preferences SP is compatible with (weak) core stability. However, such solutions are hard to compute in the case of $\EA$ preferences~\cite{dimitrov2006simple}. 
Beyond $\FA$ and $\EA$ preferences, for more general friends-oriented preferences~\cite{klaus2023core} and ASHGs~\cite{rodriguez2009strategy}, SP mechanisms guaranteeing stability have been investigated.

Some recent studies, instead of seeking stability, have concentrated on strategyproof mechanisms approximating the maximum social welfare~\cite{varricchio2023approximate}. For ASHGs in general, even when the range of agents' utilities is bounded, it was proven that a non-manipulable algorithm with a bounded approximation ratio cannot exist~\cite{flammini2021strategyproof}; the authors also provide bounded, but non-constant SP mechanisms for very restricted settings. Similarly, in the $\FA$ model, a deterministic mechanism with the approximation ratio linear in the number of agents and a randomized one with a constant approximation ratio have been provided \cite{flammini2022strategyproof}. Also, in the case of $\EA$ preferences, the best-known polynomial algorithm achieves a linear approximation in the number of agents, while a constant approximation ratio is possible when time complexity is not a concern~\cite{flammini2022strategyproof}, and this result has been proven to be asymptotically tight. Some attempts in achieving SP and bounded approximation have been made also for a proper superclass of HGs, namely, the group activity selection problem~\cite{flammini2022approximate}.

In contrast to SP, in the past few years, non-obvious manipulability has been introduced~\cite{troyan2020obvious}. This notion turned out to be a relaxation good enough to circumvent the inherent impossibility results of strategyprofness. In voting theory, non-obvious manipulability allows to bypass of a famous strong negative result stating the non-existence of a voting rule for more than two alternatives which is at the same time strategyproof and not dictatorial \cite{aziz2021obvious}. In the assignment problem, under a minor restriction the whole class of rank-minimising mechanisms, which directly optimize an objective natural for this problem, turns out to be NOM~\cite{troyan2024non}. In the problem of fairly allocating indivisible goods, replacing strategyproofness with non-obvious manipulability allows the design of a Pareto-efficient and non-dictatorial mechanism as well as a mechanism that guarantees envy-freeness up to one item~\cite{psomas2022fair}.

Since in HGs strategyproof mechanisms often fail to approximate the maximum social welfare with a constant ratio or turn out to be ineffective in the computational sense, this provides us with additional motivation to study NOM mechanisms.

\section{Preliminaries}

\paragraph{Hedonic Games and Friends Appreciation preferences.} In the classical framework of HGs, we are given a set of $n$ agents, denoted by $\agents=\set{1,\dots,n}$, and the goal is to partition them into disjoint coalitions. In other words, we aim at creating a disjoint partition $\pi=\set{C_1,\dots, C_m}$ such that $\cup_{h=1}^m C_h = \agents$ and $C_h \cap C_k= \emptyset$, for $h\neq k$. Such a partition is also called an \emph{outcome} or a \emph{coalition structure}. The \emph{grand coalition} $\GC$ is a partition consisting of exactly one coalition containing all the agents, while a {\em singleton coalition} is any coalition of size $1$. We denote by $\outcomes$ the set of all possible outcomes of the game, i.e., all possible partitions of the agents, and by $\pi(i)$ the coalition that agent 
$i$ belongs to in a given outcome $\pi\in\outcomes$. 

In HGs, the agents evaluate an outcome on the sole basis of the coalition they belong to and not on how the others aggregate.
As a result, each agent $i$ has a preference relation $\succeq_i$ over $\agents_i$, where $\agents_i$ is the family of subsets of $\agents$ containing $i$. Given $X,Y\in \agents_i$, we say that agent $i$ prefers, or equally prefers, $X$ to $Y$ whenever $X\succeq_i Y$.

In HGs with {\em friends appreciation} ($\FA$) preferences, each agent $i$ partitions the other agents into a set of friends $F_i$ and a set of enemies $E_i$, with $F_i \cup E_i = \agents \setminus \{i\}$ and $F_i \cap E_i = \emptyset$. The preferences of $i$ among coalitions in $\agents_i$ are as follows: $X\succeq_iY$ if and only if 
\begin{align*}
&\modulus{X\cap F_i}> \modulus{Y\cap F_i}
\mbox{ or } \\
&\modulus{X\cap F_i} = \modulus{Y\cap F_i} \mbox{ and }  \modulus{X\cap E_i} \leq \modulus{Y\cap E_i}\, .
\end{align*}
In other words, a coalition is preferred over another one if it contains a higher number of friends; if the number of friends is the same, the coalition with fewer enemies is preferred.

\begin{example} \label{example::instance}
Let us describe a simple instance with friends appreciation preferences: Let $\agents=\set{1,2,3}$ be the set of agents, and let $F_1=\set{2}$, $F_2=\set{3}$, $F_3=\set{2}$ and $E_1=\set{3}$, $E_2=\set{1}$, $E_3=\set{1}$ be the agents' sets of friends and enemies, respectively. 
This instance is depicted in \Cref{subfig::firstInstance}, where a directed edge from agent $i$ to agent $j$ represents $i$'s opinion of $j$; solid edges and dashed edges represent friend and enemy relations, respectively. 
\end{example}

\begin{figure}[tb]\centering 
\def \variable {1.5cm}
	\begin{subfigure}{0.45\textwidth}
		\centering
		\begin{tikzpicture}[->,>=stealth',auto,node distance=\variable,
		thick,main node/.style={circle,draw,font=\sffamily\bfseries\small}]
		\node[main node] (2) {$2$};
		\node[main node] (1) [above left of=2] {$1$};
		\node[main node] (3) [above right of=2] {$3$};
		\path[every node/.style={font=\sffamily\small}]
		(1) edge [bend right =15]  node {} (2)
		(3) edge [bend right =15] node {} (2)
		(2) edge [bend right =15] node {} (3)
		(2) edge [bend right =15, dashed]  node {} (1)
		(1) edge [bend right =15, dashed]  node {} (3)
		(3) edge [bend right =15, dashed]  node {} (1)	
		;
		\end{tikzpicture}
		\caption{An $\FA$ instance $\instance$.}
	\label{subfig::firstInstance}
	\end{subfigure}
	\begin{subfigure}{0.45\textwidth}
		\centering
		\begin{tikzpicture}[->,>=stealth',auto,node distance=\variable,
		thick,main node/.style={circle,draw,font=\sffamily\bfseries\small}]
		\node[main node] (2) {$2$};
		\node[main node] (1) [above left of=2] {$1$};
		\node[main node] (3) [above right of=2] {$3$};
		\path[every node/.style={font=\sffamily\small}]
		(1) edge [bend right =15]  node {} (2)
		(3) edge [bend right =15] node {} (2)
		(2) edge [bend right =15] node {} (3);
		\end{tikzpicture}
		\caption{The friendship graph~$\Gf(\instance)$.}
		\label{subfig::secondInstance}
	\end{subfigure}
 \caption{A $\FA$ instance and the corresponding $\Gf$. 
 Solid (resp.\ dashed) edges represent friend (resp.\ enemy) relations.
 }
\end{figure}
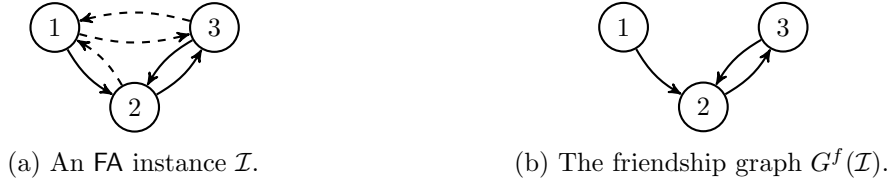

For our convenience, we shall denote by $\iIsFriend{i}=\set{j\in \agents \,\vert\, i\in F_j}$, that is, the set of agents considering $i$ a friend.

$\FA$ are a proper subclass of ASHGs, 
where each agent $i$ has a value $v_i(j)$ for every other agent $j$ and her utility for being in a given coalition $C\in\agents_i$ is $u_i(C)= \sum_{j \in C \setminus \{i\}} v_i(j)$. 
Specifically, to fit the $\FA$ preferences, the values of an agent $i$ can be defined as
\begin{align*}
v_i(j)= \begin{cases}
1,& \enspace \mbox{ if } j\in F_i, \\ -\frac{1}{n},& \enspace \mbox{ if } j\in E_i \, .
\end{cases}
\end{align*}
These values correctly encode $\FA$ preferences as the sum of the absolute values of all the enemies never exceeds the value of one friend. The valuation functions just presented were already assumed in~\cite{dimitrov2006simple,flammini2021strategyproof}. Since the utility of an agent depends only on the coalition she belongs to, we might write $u_i(\pi)$ to denote $u_i(\pi(i))$.

\begin{example}\label{example::computingUtilities}
Let us consider the instance described in \Cref{example::instance} and the partition $\pi=\set{\set{1,2,3}}$. Then, $u_1(\pi)= v_1(2) +v_1(3) = 1 - \frac{1}{3} = \frac{2}{3}$. Similarly, since the utility of an agent depends only on the number of friends/enemies in her coalition, $u_2(\pi)=u_3(\pi)=\frac{2}{3}$.
\end{example}

An $\FA$ instance $\instance$ is given by a set of agents $\agents$ and a set of friends $F_i$, for each $i\in\agents$. Alternatively, $\instance = (\agents, \set{v_i}_{i\in\agents})$, where $v_i$ is the valuation of $i$ for the other agents representing her $\FA$ preferences. For simplicity, we might also write $\instance = (\set{v_i}_{i\in\agents})$.

\paragraph{Social Welfare and Optimum.}
One of the challenges in HGs is to determine which outcomes maximize the overall happiness of the agents measured by the \emph{social welfare} (SW). Specifically, in an HG instance $\instance$ the social welfare of a partition $\pi$ is given by
\[
\SW^\instance(\pi)= \sum_{i\in\agents} u_i(\pi) \, .
\]
When the instance is clear from the context we simply write $\SW$.
We call \emph{social optimum}, or simply \emph{optimum}, any outcome $\OPT$ in $\argmax\limits_{\pi\in\outcomes}\SW(\pi)$ and denote by $\opt$ the value $\SW(\OPT)$. When considering a coalition $C$, to denote $\sum_{i\in C} u_i(C)$ we write $\SW(C)$.

\paragraph{Graph Representation}
A very convenient representation of an $\FA$ instance $\instance$ is by means of a directed and unweighted graph where the agents of the instance are the vertices. With $E_i$ being $\agents\setminus \set{F_i\cup \set{i}}$, we represent only friendship relationships through directed edges: if $\set{i,j}$ is an edge of this graph, it means $j\in F_i$; if such an edge does not exist, we have $j\in E_i$. We call this graph the {\em friendship graph} of $\instance$ and we denote it by $\Gf(\instance)= (\agents, F)$, where $F=\set{\set{i,j} \,\vert\, j\in F_i}$. If the instance is clear from the context we simply write $\Gf$. The friendship graph of the instance described in \Cref{example::instance} is shown in \Cref{subfig::secondInstance}.
Moreover, we denote by $\neigh(i)$, for $i\in\agents $, the weakly connected neighborhood of $i$, that is, $\neigh(i)= F_i \cup F^{-1}_i$. We denote by $\delta(i)$ the size of $\neigh(i)$. We further extend the definition of weekly connected neighborhood to a subset of agents $X \subseteq \agents$, specifically $\neigh(X)= \cup_{i\in X} \neigh(i)$.

\paragraph{Strategyproofness and Non-obvious Manipulability.}
The sets $F_i$ and $E_i$ might be private information of the agent $i$; therefore, to compute the outcome we need to receive this information from the agents. Let us denote by $\declared=(d_1, \dots, d_n)$ the agents' {\em declarations} vector, where $d_i$ contains the information related to agent $i$. We assume direct revelation, and hence $d_i(j)\in \set{1, -\frac{1}{n}}$ represents the value $i$ declared for an agent $j$.
We denote by $\declarations$ the space of feasible declarations $\declared$. For our convenience, we denote by $\declared_{-i}$ the agents' declarations except the one of $i$, by $\declarations_{-i}$ the set of all feasible $\declared_{-i}$, and by $\declarations_i$ the feasible declarations for $i$.

In this setting, the natural challenge is to design algorithms, a.k.a.\ \emph{mechanisms}, inducing truthful behavior of the agents. We shall denote by $\mech$ a mechanism and by $\mech(\declared)$ the output of the mechanism -- a partition upon the declaration $\declared$ of the agents. 

The agents might be strategic, which means, an agent $i$ could declare $d_i\neq t_i$, where $t_i\in\declarations_i$ is the real information of agent $i$, also called her \emph{real type}. For this reason, the design of mechanisms preventing manipulations is fundamental. The most desirable characteristic for such kind of mechanisms is \emph{strategyproofness}.

\begin{definition}[Strategyproofness and Manipulability]\label{def:SP}
    A mechanism $\mech$ is said to be \emph{strategyproof} (SP) if for each $i\in\agents$ having real type $t_i$, and any declaration of the other agents $\declared_{-i}$
    \begin{equation}\label{eq:strategyproofness}
         u_i(\mech(t_i, \declared_{-i})) \geq u_i(\mech(d_i, \declared_{-i}))
    \end{equation}
    holds true for any possible false declaration $d_i\neq t_i$ of agent $i$.

    In turn, a mechanism is said to be \emph{manipulable} if there exists an agent $i$, a real type $t_i$ and declarations $\declared_{-i}$ and $d_i\neq t_i$ such that \Cref{eq:strategyproofness} does not hold. Then, such a $d_i$ is called a \emph{manipulation}.
\end{definition}
Since SP mechanisms may be quite inefficient w.r.t.\ the truthful $\opt$, we aim to understand if mechanisms satisfying milder conditions lead to more efficient outcomes. 
Considering that $i$ might be unaware of which are the declarations $\declared_{-i}$ of the other agents, she could not be able to determine a manipulation without knowing $\declared_{-i}$. Thus,  we next consider a relaxation of SP where an agent $i$ decides to misreport her true values only if it is clearly profitable for her. 
Given a mechanism $\mech$, let us denote by 
$
\Pi_i(d_i,\mech)= \set{\mech(d_i, \declared_{-i}) \,\vert\,\declared_{-i}\in \declarations_{-i}}\, ,
$
the space of possible outcomes of $\mech$ given the declaration $d_i$ of $i$.  Notice the space $\Pi_i(d_i,\mech)$ is finite.
\begin{definition}[Non-obvious Manipulability]\label{def:NOM}
    A mechanism $\mech$ is said to be \emph{non-obviously manipulable} (NOM) if for every $i\in\agents$, real type $t_i$, and any other declaration $d_i$ the following hold true:
\medskip
\\
\textbf{Condition 1:} $\max\limits_{\pi\in \Pi_i(t_i,\mech)} u_i(\pi) \geq \max\limits_{\pi\in \Pi_i(d_i,\mech)} u_i(\pi)$
\smallskip
\\
\textbf{Condition 2:} $\min\limits_{\pi\in \Pi_i(t_i,\mech)} u_i(\pi) \geq \min\limits_{\pi\in \Pi_i(d_i,\mech)} u_i(\pi)$
\medskip
\\
If there exist $i$, $t_i$, and $d_i$ such that Condition 1 or 2 is violated, then, $\mech$ is \emph{obviously manipulable} and $d_i$ is an \emph{obvious manipulation}.
\end{definition}

In other words, a mechanism $\mech$ is NOM for an agent $i$, neither the best nor the worst possible outcome when declaring $t_i$ can be strictly worse than the best or the worst outcome attainable when declaring $d_i$ (worst/best outcomes and their comparisons are determined according to her truthful preferences).
In contrast, the strategyproofness of a mechanism $\mech$ ensures that for every $\declared_{-i}$, including the ones inducing the best/worst case outcome of $\Pi_i(t_i, \mech)$, is not strictly convenient to misreport; therefore, SP $\Rightarrow$ NOM.

In what follows, we always denote by $t_i$ the real type of $i$ and by $e_i=\modulus{E_i}$ and $f_i=\modulus{F_i}$, where $E_i$ and $F_i$ are the truthful set of friends and enemies of $i$, respectively.

\subsection{Preliminary Results on Optimum Outcomes}

In this section, we discuss in detail some useful properties of optimal outcomes putting particular attention on specific graph structures for the friendship graph.
Let us start by observing that to compute the social welfare of a coalition it suffices to know its size and the number of friendship relationships within the coalition.

\begin{lemma}[From \cite{flammini2022strategyproof}]\label{prop-utility-coaltion}
For any $C\subseteq \agents$ of size $c$ and containing $f_C$ friendship relations, $\SW(C)= f_C\cdot \left(1+\frac{1}{n}\right)-\frac{c(c-1)}{n}$.
\end{lemma}

\begin{example}\label{ex::star}
Consider, for example, an $\FA$ instance where $\Gf$ is a star whose edges are directed from its center $i$ towards the leaves, that is, $F=\set{\set{i,j} \,\vert\, j\in\agents \setminus\set{i}}$. For such an instance, if we put the agents together in the grand coalition, we have $\SW(\GC)= \frac{n-1}{n}$. 
Clearly, in any optimum $\pi^*$, any node that is not in the same coalition as $i$ must be in a singleton coalition, otherwise, the $\SW$ of its coalition would be negative. Let $C$ be the coalition of $i$ in $\pi^*$, $\SW(\pi^*)=\SW(C)= (c-1)\cdot \left(1+\frac{1}{n}-\frac{c}{n}\right)$, which is maximized at $c= \frac{n+2}{2}$, for even $n$, and at $c= \frac{n+2}{2} \pm \frac{1}{2}$, otherwise. Note that the optimality does not depend on the edges direction.
\end{example}

\Cref{ex::star} shows that when the graph is particularly sparse it is more convenient to split weakly connected components rather than put all agents with all their friends.
In turn, when there exists a cluster of nodes $C\subset \agents$, that is, $C$ is a bidirectional clique in $\Gf$, whose nodes are weakly connected only to another node $i\in\agents \setminus C$, it is never convenient to split the agents in $C$. We call $C$ an {\em almost isolated clique} with {\em hinge node} $i$. We next formalize how almost isolated cliques place in an optimum outcome.

\begin{restatable}{lemma}{almostIsolatedOPT}\label{lemma:almostIsolatedOPT}
If $C$ is an almost isolated clique in $\Gf$ with hinge node $i$, for any optimal partition $\pi^*$ there exists $C' \in \pi^*$ such that $C \subseteq C'$. Furthermore, if $i\not\in C'$ then $C'=C$.
\end{restatable}

Next, we consider more involved structures for $\Gf$ and explain how their optimal outcomes look like.
\begin{definition}[Octopus Graph]
Given an agent $i$ and $H\subseteq \agents\setminus\set{i}$, $\Gf=(\agents, F)$ is an $i$-centered {\em octopus graph} with the head $H$ if:
\begin{itemize}
    \item $H$ is a bidirectional clique in $\Gf$;
    \item for each $j \in H$, $\set{j,i}\in F$;
    \item for each $j \in  \agents \setminus{i}$ and $k \in \agents\setminus\left(\{i\} \cup H\right)$, none of $\set{j, k}$, $\set{k, j}$, $\set{k, i}$ belongs to $F$ while $\set{i, j}$ may belong to $F$.
\end{itemize}
\end{definition}
A picture of an $i$-centered octopus graph is given in \Cref{fig::octopus}.

\begin{restatable}{lemma}{octopusOPT}\label{lemma:octopusOPT}
Let $\Gf$ be an $i$-centered octopus graph with head $H$. If $|H| \geq \left\lceil \frac{n}{2} \right\rceil$, there exists a unique optimum consisting of the coalition $H \cup \set{i}$ and remaining agents put in singleton coalitions.
\end{restatable}
\begin{proof}[Sketch]
Let us start by noticing that $H$ is an almost isolated clique with hinge $i$. By \Cref{lemma:almostIsolatedOPT}, in the social optimum, all agents from $H$ will end up in the same coalition. Moreover, for each agent $k \in \agents\setminus\left(\{i\} \cup H\right)$, $k$ can be weakly connected only to $i$, so, if $k$ is not in the same coalition as $i$, then $k$ forms a singleton coalition. This leaves us with three possible types of optimal partition:
\begin{enumerate}
    \item[$\pi^1$] where $H$ and $i$ are in the same coalition, while all remaining agents are in singletons;
    \item[$\pi^2$] where agents from $H$ form a coalition, while $i$ is in a different coalition together with some $C_1 \subseteq \agents\setminus\left(\{i\} \cup H\right)$, and remaining agents are in singletons;
    \item[$\pi^3$] where $H$, $i$ and some $C_2 \subseteq\agents\setminus\left(\{i\} \cup H\right)$ form one coalition, while all other agents are in singletons.
\end{enumerate} 
Let us compare $\pi^1$ and $\pi^2$. The number of positive relationships between $H$ and $i$ is at least $\modulus{H}$ while the positive connections between $i$ and $C_1$ are at most $\modulus{C_1}$. Since $\modulus{C_1}\leq n - \modulus{H} -1 \leq  n- \ceil{\frac{n}{2}} -1 < \modulus{H}$, it is strictly more convenient to put $i$ in coalition with $H$ rather than with $C_1$, showing $\SW(\pi^1) > \SW(\pi^2)$.

Let us compare $\pi^1$ and $\pi^3$. 
Consider the coalition $H\cup C_2 \cup \set{i}$. The friendship relationships between $H\cup \set{i}$ and $C_2$ are only the ones in $F_i\cap C_2$, which are at most $\modulus{C_2}$.  So, if we remove and split $C_2$ into singletons, the loss in social welfare will be of at most
$\modulus{C_2}\cdot\left(1 + \frac{1}{n}\right) - \frac{2\cdot(\modulus{H}+1)\cdot\modulus{C_2}}{n} \leq \modulus{C_2}\cdot \frac{n+1 - 2\cdot\frac{n}{2}-2}{n} <0$, where the first inequality holds as $\modulus{H}\geq \frac{n}{2}$. Therefore, $\SW(\pi^1) > \SW(\pi^3)$. 
\end{proof}

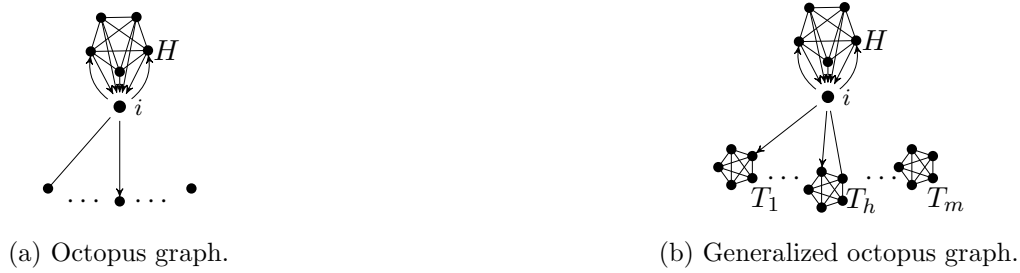
\begin{figure}[tb]
	\centering
	\def \variable {0.25cm}
        \def \var {0.4cm}
	\pgfdeclarelayer{bg}    
	\pgfsetlayers{bg,main} 

\begin{subfigure}[t]{0.42\textwidth}
\centering
	\begin{tikzpicture}[-,>=stealth',shorten >=0pt,auto,node distance=\variable,
	thick,main node/.style={circle,draw,fill, minimum size=1pt, scale=0.3}]

	\foreach \numb in {1,2,3,4,5} {
		\node[main node] (\numb) at ({72*\numb +55}:\var) {};
	}
 \foreach \numb in {6,7,8,9,10} {
    \node[main node] (\numb) [yshift= -7cm, color=white] at ({72*\numb +108}:\variable) {};
	}
    
	\node[main node] (i) [below of = 3, yshift= -1.3cm , color=white, minimum size=35pt] {};
    \node[main node] (i2) [below of = 3, yshift= -1.3cm ,  circle,draw,fill, minimum size=13pt]  {};
	\node () [ right of =i2, font=\footnotesize] {$i$};
    \node[main node] (a) [below of = 3, yshift= -5.5cm] {};
    \node[main node] (b) [below left of = 3, yshift= -5cm, xshift= -3 cm ] {};
    \node[] () [right of = a, xshift=0.2cm] {$\cdots$};
    \node[main node] (c) [below right of = 3, yshift= -5cm, xshift= 3cm ] {};
    \node[] () [left of = a, xshift=-0.2cm] {$\cdots$};
    \node[] () [ right of =4] {$H$};

 \path[every node/.style={font=\sffamily\small}, edge/.style ={dashed}]

    \foreach \B in {1,2,3,4,5}{   
		
		(\B) edge [line width=0.3pt, ->] node {} (i)
	}
	\foreach \A in {1,2,3,4,5}{
		\foreach \B in {1,2,3,4,5}{   
			\ifnum \A<\B
			(\A) edge [line width=0.3pt] node {} (\B)
			\fi
		}
	}
 (i) edge [line width=0.3pt, ->, bend right] node {} (4)
(i) edge [line width=0.3pt, ->, bend left] node {} (2)
(i) edge [line width=0.3pt, ->] node {} (a)
(i) edge [line width=0.3pt, -] node {} (b)

	;
\end{tikzpicture}
\caption{Octopus graph.}
\label{fig::octopus}
\end{subfigure}
\hfill
\begin{subfigure}[t]{0.42\textwidth}
\centering
	\begin{tikzpicture}[-,>=stealth',shorten >=0pt,auto,node distance=\variable,
	thick,main node/.style={circle,draw,fill, minimum size=1pt, scale=0.3}]

	\foreach \numb in {1,2,3,4,5} {
		\node[main node] (\numb) at ({72*\numb +55}:\var) {};
	}
    
	\node[main node] (i) [below of = 3, yshift= -1.3cm , color=white, minimum size=35pt] {};
        \node[main node] (i2) [below of = 3, yshift= -1.3cm ,  circle,draw,fill, minimum size=13pt]  {};
	\node[] () [ right of =i2, font=\footnotesize] {$i$};
  
	\foreach \numb in {6,7,8,9,10} {
    \node[main node] (\numb) [yshift= -7cm] at ({72*\numb +108}:\variable) {};
	}

	\foreach \numb in {11,12,13,14,15}{
    \node[main node] (\numb) [yshift= -6cm, xshift= 4cm] at ({72*\numb +108}:\variable) {};
	}
 
 	\foreach \numb in {16,17,18,19,20}{
    \node[main node] (\numb) [yshift= -6cm, xshift= -4cm] at ({72*\numb +108}:\variable) {};
	}

 	\node[] () [ right of =4] {$H$};
	\node[] () [ right of =8] {$T_h$};
        \node[] () [ right of =9, xshift= 0.3cm] {$\cdots$};
        \node[] () [ left of =9, xshift= -0.5cm] {$\cdots$};
	\node[] () [below right of =13, yshift=-0.1cm ] {$T_m$};
	\node[] () [below right of =18, yshift=-0.1cm ] {$T_1$};
 \path[every node/.style={font=\sffamily\small}, edge/.style ={dashed}]
	\foreach \A in {1,2,3,4,5}{
		\foreach \B in {1,2,3,4,5}{   
			\ifnum \A<\B
			(\A) edge [line width=0.3pt] node {} (\B)
			\fi
		}
	}

	\foreach \A in {6,7,8,9,10}{
		\foreach \B in {6,7,8,9,10}{    
			\ifnum \A<\B
			(\A) edge [line width=0.3pt] node {} (\B)
			\fi
		}
	}

        \foreach \A in {11,12,13,14,15}{
		\foreach \B in {11,12,13,14,15}{    
			\ifnum \A<\B
			(\A) edge [line width=0.3pt] node {} (\B)
			\fi
		}
	}	
        \foreach \A in {16,17,18,19,20}{
		\foreach \B in {16,17,18,19,20}{    
			\ifnum \A<\B
			(\A) edge [line width=0.3pt] node {} (\B)
			\fi
		}
	}	
    \foreach \B in {1,2,3,4,5}{   
		
		(\B) edge [line width=0.3pt, ->] node {} (i)
	}

 (i) edge [line width=0.3pt, ->, bend right] node {} (4)
(i) edge [line width=0.3pt, ->, bend left] node {} (2)
(i) edge [line width=0.3pt, ->] node {} (19)
(i) edge [line width=0.3pt, -] node {} (9)
(i) edge [line width=0.3pt, ->] node {} (10)
	;
\end{tikzpicture}
\caption{Generalized octopus graph.}
\label{fig::generalizedOctopus}
\end{subfigure}
\caption{Octopus graph structures having center $i$. Undirected edges represent bidirectional edges in $\Gf$.}
\label{fig::graphStructures}
\end{figure}

We now further generalize the definition of octopus graph. 

\begin{definition}[Generalized Octopus Graph]
Given an agent $i\in\agents$ and $\set{H,T_1, \dots , T_m}$, a disjoint partition of $\agents\setminus\set{i}$; $\Gf=(\agents, F)$ is an $i$-centered {\em generalized octopus graph} with the head $H$ and the tentacles $T_1, \dots , T_m$ if, for each $l \in [m]$:
\begin{itemize}
    \item $H$ and $T_l$  are bidirectional cliques in $\Gf$;
    \item for each $j \in H$, $\set{j,i}\in F$;
    \item for each $j \in \agents\setminus \set{T_l \cup \set{i}}$ and $k \in T_l$ none of $\set{j, k}$, $\set{k, j}$, $\set{k, i}$ belongs to $F$ while $\set{i, j}$ may belong to $F$.
\end{itemize}
\end{definition}

Given a node $i$, we denote by $\goct(i)$ the set of all possible $i$-centered generalized octopus graphs. In \Cref{fig::generalizedOctopus}, we draw an example of an $i$-centered generalized octopus graph. Let us note that an octopus graph is a generalized octopus graph having all tentacles of size $1$. Furthermore, if $\Gf(d_i, \declared_{-i})\in \goct(i)$, for any declaration $d'_i\in\declarations_i$, $\Gf(d'_i, \declared_{-i})\in \goct(i)$; in fact, in the definition there is no constraint on how the set of friends of $i$ should be.

\begin{lemma}\label{lemma:genOctopusOPTtentacles}
Let $\Gf$ be a generalized $i$-centered octopus graph with head $H$ and tentacles $T_1, \dots, T_m$. If $\pi^*=\set{H\cup\set{i}, T_1, \dots, T_m}$ is an optimum outcome, then $|H| \geq \max_{l\in[m]}\frac{\modulus{F_{i} \cap  T_l}}{|T_l|}\cdot \frac{n+1}{2} - 1$.
\end{lemma}

\begin{proof}
Let us set $g_l=\modulus{F_{i} \cap  T_l}$.
Consider the partition $\pi$ obtained from $\pi^*$ by merging the coalitions $H\cup\set{i}$ and $T_l$. Being $\pi^*$ optimum,
$
 \SW(\pi) - \SW(\pi^*) = g_l- \frac{2|T_l| - g_l}{n} - \frac{2|T_l|\cdot|H|}{n} \leq  0 \, .
$
Therefore, 
$
|H| \geq \frac{g_l(n+1)}{2|T_l|} - 1$, 
and hence it holds for the maximum.
\end{proof}

Next, we provide an interesting connection between the optimum for arbitrary instances and the one for octopus graphs.

\begin{lemma}\label{lemma:lookOnlyOctopus}
If $\pi^*$ is an optimum partition for $\instance=(d_i, d_{-i})$, then, $\pi^*$ is also optimum for $\instance'=(d_i, d'_{-i})$, where $\Gf(\instance')\in\goct(i)$ 
with the head $\pi(i) \setminus \set{i}$ and the remaining coalitions of $\pi$ being the tentacles.
\end{lemma}
In a nutshell: Strengthening friendships within and enmities between the coalitions of $\pi$ maintains optimality.

\begin{proof}
Let us denote by $F^\instance$ and $F^{\instance'}$  the friendship relationship in $\instance$ and $\instance'$, respectively. 
We set $p=\modulus{F^{\instance'} \setminus F^{\instance}}$ and $q= \modulus{F^{\instance}\setminus F^{\instance'}}$, which denote the number of added and removed edges when transforming $\Gf(\instance)$ into $\Gf(\instance')$. 
Then, we have $\SW^{\instance'}(\pi) - \SW^{\instance}(\pi) = p\cdot\left(1+\frac{1}{n}\right)$ since the edges we added are within coalitions and the edges we removed are between coalitions. Consider now an outcome $\pi' \neq \pi$ and let $p'$ and $q'$ be the number of the added and removed edges, whose endpoints are in the same coalition in $\pi'$, when transforming $\Gf(\instance)$ into $\Gf(\instance')$.
In this case, being $p'\leq p$,
\[
\SW^{\instance'}(\pi') - \SW^{\instance}(\pi') = p'\cdot\left(1 + \frac{1}{n}\right) - q'\cdot\left(1 + \frac{1}{n}\right) \leq p\cdot\left(1 + \frac{1}{n}\right) \, .
\]
Therefore,
$\SW^{\instance'}(\pi) - \SW^{\instance'}(\pi') \geq \SW^{\instance}(\pi)  - \SW^{\instance}(\pi') \geq 0$, for any $\pi'\in \Pi$,
showing the optimality of  $\pi$ for $\instance'$.
\end{proof}
\section{An optimal and NOM mechanism}
In~\cite{flammini2022strategyproof}, it has been shown that no strategyproof mechanism can have an approximation better than $2$. In contrast, we next show there is a way to simultaneously guarantee optimality and NOM.

Let us first introduce our optimal mechanism.
\begin{mechanism}\label{mech:OPT}
    It always returns an optimum partition with the smallest number of coalitions, ties among partitions with the same number of coalitions are broken arbitrarily.
\end{mechanism}

Clearly, the mechanism returns an optimal partition; however, we need to select a specific optimum, if the instance admits more than one optimum, in order to show the following:

\begin{theorem}\label{thm:optIsNOM}
   Mechanism \ref{mech:OPT} is NOM.
\end{theorem}

To show the theorem, we at first need to understand which are the worst/best outcomes for $i$ in $\Pi_i(t_i, \ref{mech:OPT})$, the space of possible outcomes of \ref{mech:OPT} when $i$ reports $t_i$. We will then compare their utility for $i$ with the one of the worst/best outcomes in $\Pi_i(d_i, \ref{mech:OPT})$ for any other feasible $d_{i}$.
Recall that $e_i=\modulus{E_i}$ and $f_i=\modulus{F_i}$ are the sizes of the truthful set of friends and enemies of $i$, respectively. 

\begin{lemma}\label{lemma:bestWorstCaseOptMech}
For any agent $i\in\agents$, among the outcomes in $\Pi_i(t_i, \ref{mech:OPT})$:
\begin{enumerate}
     \item\label{item:bestCase} any outcome where $i$ is put in the coalition $\set{i}\cup F_i$ maximizes the utility of $i$;
     \item\label{item:worstCase} any outcome where $i$ is put in a coalition with $E_i$ and $\max\set{\ceil{\frac{n}{2}}-e_i, 0}$ friends minimizes the utility of $i$.
\end{enumerate}
\end{lemma}
\begin{proof}
    We first show \ref{item:bestCase}. Given $t_i$, let us set $\declared_{-i}$ in such a way that $\set{i}\cup F_i$ is a bidirectional clique in the corresponding friendship graph $\Gf$ while all remaining agents are isolated nodes. For such instance, there is a unique optimum and it is attained by the partition where $\set{i}\cup F_i$ is the only non-singleton coalition. Clearly, no partition can guarantee $i$ strictly higher utility, therefore \ref{item:bestCase} follows.

    We now show \ref{item:worstCase} and distinguish between $e_i \geq \ceil{\frac{n}{2}}$ and $e_i < \ceil{\frac{n}{2}}$.

    If $e_i \geq \ceil{\frac{n}{2}}$, consider a declaration of the others $\declared_{-i}$ such that the friendship graph $\Gf(t_i,\declared_{-i})$ is an $i$-centered octopus graph with the head $E_i$. By~\Cref{lemma:octopusOPT}, there exists a unique optimum, which is therefore the output of \ref{mech:OPT}, and $E_i\cup\set{i}$ is one of its coalitions. Since there exists no coalition where $i$ gets strictly lower utility, we can conclude it is the worst possible outcome of \ref{mech:OPT} for $i$.

    Assume now $e_i < \ceil{\frac{n}{2}}$. We first show there exists $\declared_{-i}$ such that in $\ref{mech:OPT}(t_i,\declared_{-i})$, $i$ is in a coalition  with the agents $C\subset\agents\setminus \set{i} $ with $f'= \ceil{\frac{n}{2}} - e_i$ many friends. In fact, in this case, $\modulus{C} \geq \ceil{\frac{n}{2}}$ and we can build as before $\declared_{-i}$ so as the friendship graph $\Gf(t_i,\declared_{-i})$ is an $i$-centered octopus graph with the head $C$ and apply \Cref{lemma:octopusOPT}. To conclude this is the worst outcome for $i$ declaring $t_i$, we next show there is no outcome where $i$ has strictly less friends.

    Assume $i$ is put into a coalition  together with $C\subset \agents\setminus\set{i}$ where  $f'= |C\cap F_i| < \ceil{\frac{n}{2}} - e_i$. Note that having the same number of friends and a greater number of enemies is not possible.

    Assume such $C\cup\set{i}$, having size $c=\modulus{C}$, is a coalition in the social optimum, say $\pi^*= \set{C\cup\set{i}, T_1, \dots , T_m}$, for a game instance $(d_i, \declared_{-i})$. Let us then select $\declared'_{-i}$ so that $\Gf(d_i, \declared'_{-i})$ is an $i$-centered generalized octopus graph having head $C$ and tentacles $T_1, \dots , T_m$. From the proof of \Cref{lemma:lookOnlyOctopus} $\pi^*$ is optimum also for $(d_i, \declared'_{-i})$.  By \Cref{lemma:genOctopusOPTtentacles}, $\pi^*$ is optimum only if $c > \max_l\left\{\frac{|T_l \cap F_i|}{|T_l|}\right\}\cdot \frac{n+1}{2} - 1$. We notice that, \Cref{lemma:genOctopusOPTtentacles} the last inequality should be $\geq$; however, by adjusting the proof of the lemma the inequality turns out to be strict for the mechanism \ref{mech:OPT}, as it selects an optimal outcome with the lowest number of coalitions (see \Cref{lemma:genOctopusOPTtentaclesForMech} in the appendix). 

    Now,  
    $\max_l\left\{\frac{|T_l \cap F_i|}{|T_l|}\right\} \geq \frac{\sum_{l=1}^m |T_l \cap F_i|}{\sum_{l=1}^m |T_l|} = \frac{\left|\left(\bigcup_{l\in [m]}T_l\right) \cap F_i\right|}{\left|\bigcup_{l\in [m]}T_l\right|} = \frac{|T \cap F_i|}{|T|}$, where $T = \bigcup_{l\in [m]}T_l$. This implies, $c > \frac{|T \cap F_i|}{|T|}\cdot \frac{n+1}{2} - 1$. 
    Therefore,  $c > \frac{\ceil{\frac{n}{2}} - e_i - f'}{n- c - 1} \cdot \frac{n+1}{2} - 1$ as $T= \agents\setminus\set{C\cup \set{i}}$ and $\modulus{T \cap F_i}= \ceil{\frac{n}{2}} - e_i - f'$, and hence
    \[
    f' > \ceil{\frac{n}{2}} - e_i - (c+1)(n-c-1)\cdot\frac{2}{n+1} \, .
    \]

Now, recall that $f' \leq  \ceil{\frac{n}{2}} - e_i - 1$, thus,
\begin{gather*}
\left\lceil\frac{n}{2} \right\rceil - e_i - 1 > \ceil{\frac{n}{2}} - e_i - (c+1)(n-c-1)\cdot\frac{2}{n+1} \\
\left\lceil\frac{n}{2} \right\rceil - 1 > n- 1 - (c+1)(n-c-1)\cdot\frac{2}{n+1} \\
2c^2 + 2(2-n)c + n^2 - n+ 2 - \left\lceil\frac{n}{2} \right\rceil \cdot(n+1) < 0 \, .
\end{gather*}

The discriminant with respect to $c$ of the left side is equal to $4\left(-n^2 -2n+ 2\left\lceil\frac{n}{2}\right\rceil\cdot(n+1)\right)$. When $n$ is even, this expression equals $-4n$, therefore, the inequality above does not have a solution. If $n$ is odd, $\left\lceil\frac{n}{2}\right\rceil = \frac{n+1}{2}$, and the inequality has the following solution:
$
\frac{n-3}{2} < c < \frac{n-1}{2} \, .
$
However, this interval does not contain integer numbers.
In conclusion, such $C$ cannot exist when$f' < \ceil{\frac{n}{2}} - e_i$.
\end{proof}
We are now ready to show \ref{mech:OPT} is NOM.
\begin{proof}[Proof of \Cref{thm:optIsNOM}]
We prove both Condition 1 and 2 of the definition of NOM hold true for any agent $i$. 

    {\em Condition 1.}
    This condition is the simplest to show.
    By \Cref{lemma:bestWorstCaseOptMech}, if $i$ truthfully reports, the best outcome is attained by a partition where $i$ is in a coalition with all her friends and no enemies. This coalition provides her the highest possible utility, so, Condition 1 holds true as no misreport can guarantee a strictly higher utility. 

    {\em Condition 2.}
    To understand the worst-case scenario we will make a case distinction depending on the number of enemies $i$ has.

    If $e_i \geq \ceil{\frac{n}{2}}$, for any $d_i\in \declarations_i$, consider a declaration of the others $\declared_{-i}$ such that the friendship graph $\Gf$ for the resulting instance $(d_i, \declared_{-i})$ is an $i$-centered octopus graph with the head $H= E_i$. By \Cref{lemma:octopusOPT}, for this instance, \ref{mech:OPT} outputs a partition containing $E_i\cup\set{i}$. This is the worst possible coalition for $i$ according to her truthful declaration $t_i$. Since for any possible reporting of $i$ there always exists a declaration of the others such that $i$ ens up in a coalition together with all her enemies, Condition 2 is satisfied.

    Assume now $e_i < \ceil{\frac{n}{2}}$. Let us choose $X\subseteq F_i$ such that $\modulus{E_i\cup X}= \ceil{\frac{n}{2}}$.
    We first show that, regardless of the declaration of $i$, there always exists an instance where the unique optimum has $E_i \cup X\cup \set{i}$ as a coalition. Given $d_i\in\declarations_i$, let us choose $\declared_{-i}$ in such a way that $\Gf(d_i, \declared_{-i})$ is an $i$-centered octopus graph with the head $H= E_i \cup X$. Being $\modulus{H}\geq \ceil{\frac{n}{2}}$, by \Cref{lemma:octopusOPT} we have that, in the unique optimum for this instance, $H\cup\set{i}$ is a coalition and remaining agents form singleton coalitions. By \Cref{lemma:bestWorstCaseOptMech}, such an outcome is the worst possible when declaring $t_i$, therefore, misreporting cannot improve the worst-case. This concludes our proof.
\end{proof}
\section{Computing the optimum is NP-hard}
Despite the existence of an optimum and NOM mechanism, in this section we show that computing an optimum partition is NP-hard.

\begin{theorem}\label{thm:OPThard}
   For $\FA$ preferences, computing the optimum is NP-hard.
\end{theorem}

We prove \Cref{thm:OPThard} with a reduction from the \threePartition\ problem, which can be formulated as follows:
{
\begin{center}
\fbox{
\begin{varwidth}{0.95\linewidth}
\noindent {\bf \threePartition\ problem}

\smallskip
\noindent\textit{Input:} A ground set $\set{x_1, x_2, \dots , x_{3m}}$ of $3m$ elements such that 
\begin{enumerate}[label=(\roman*)]
    \item $\sum_{h=1}^{3m} x_h = mT$, for some $T>0$;
    \item $x_h\in \mathbb{N}$, for each $h\in [3m]$;
    \item $\frac{T}{4} < x_h < \frac{T}{2}$, for each $h\in [3m]$. \label{item:condition}
\end{enumerate}

\smallskip
\noindent\textit{Question:} Does there exist a partition of the ground set into $m$ disjoint subsets $S_1, \dots, S_m$ such that, for every $k\in [m]$,
\\
$S_k=\set{s_k^1, s_k^2, s_k^3}$ and  $s_k^1 + s_k^2 + s_k^3 = T$?
\end{varwidth}
}
\end{center}
}

Let us note that in the standard formulation of \threePartition, condition~(iii) is usually not required, however, the problem remains strongly NP-hard even under such a condition~\cite{garey1979computers}.
Moreover, condition~(iii) also implies that for any $S \subseteq \{x_1, x_2, \dots, x_{3m}\}$ if $\sum_{x\in S} x = T$, then $|S| = 3$. 
Consequently, any partition into subsets, each having sum $T$, is a partition into triples.

\paragraph{Reduction.}
Given a \threePartition instance, let us construct the friendship graph $\Gf$ representing the corresponding $\FA$ instance. 

\noindent {\bf Element-cliques:} Each of these cliques represents a specific element in the ground set of the \threePartition\ instance. Specifically, for every $h\in[3m]$, we create a bidirectional clique $K^h$ of size $x_h$. 

\noindent {\bf Set-cliques:} We create $m$ bidirectional cliques $K^1_{X}, \dots, K^m_{X}$ each one being of size $X = 4m^2T$.  The choice of $X$ is made in such a way that we can use the cliques $K^1_{X}, \dots, K^m_{X}$ to interpret a coalition in an optimum partition, for the $\FA$ instance, as a set in the partition of the ground set for the \threePartition\ instance. 

\noindent {\bf Connections between cliques:} We add $x_h$ bidirectional edges between $K^h$ and each $K^k_X$ in such a way there is exactly one bidirectional edge between each vertex of $K^h$ and some node in $K^k_X$. Since $|K^h| = x_h < X$, this is always possible.

\smallskip
Notice that the number of agents is $n = \sum_{h=1}^{3m} x_h + mX = mT + 4m^3T$; thus, with \threePartition\ being strongly NP-hard the correctness of the reduction proves the NP-hardness of our problem.
\paragraph{The optimum in the reduced instance.}
As a first step to prove \Cref{thm:OPThard},
we need to understand the structure of a socially optimum outcome for the reduced instance. Thanks to a number of auxiliary lemmas, deferred to the appendix, we can conclude the following:

\begin{proposition}\label{thm:mainThmReduction}
In the reduced instance, any socially optimum partition $\pi^*$ is made by exactly $m$ coalitions and, for each $C\in\pi^*$,
\begin{enumerate}[label=(\alph*)]
    \item there exists a unique $k\in[m]$ such that $K^k_X\subseteq C$, and 
    \item for every $h\in [3m]$, either $K^h\subseteq C$ or $K^h\cap C=\emptyset$.
\end{enumerate}

\end{proposition}

To give an intuition, we chose $X$ sufficiently large so that putting two cliques of size $X$ in the same coalition would introduce too many negative relationships. Moreover, the positive connections between a clique (of both types, $K^k_X$ or $K^h$) are particularly sparse if compared to their size, this guarantees that these cliques will be put in the same coalition in any optimum outcome.

\Cref{thm:mainThmReduction} turns out to be very helpful in restricting the outcomes that may possibly be optimum: in an optimum outcome, there are exactly $m$ coalitions $C_1, \dots, C_m$ each one consisting of an $K^k_X$ and possibly some of the cliques $\{K^1, \dots, K^{3m}\}$. 
Let us denote by $\Sigma$ the set of such partitions and let us make some important observations about any $\pi\in \Sigma$ which will help us to establish the optimum social welfare. Let us count the number of friendship and enemy relationships within coalitions of $\pi$:
\begin{enumerate}
    \item The cliques never get split, and hence, inside the coalitions of $\pi$ there are always exactly $\alpha = mX(X-1) +\sum_{h = 1}^{3m} x_h(x_h-1)$ friendship relations between the members of the cliques;
    \item any element-clique is in a coalition with exactly one set-clique $K^k_X$, thus, the total number of the friendship relations between element- and set-cliques within the coalitions of $\pi$ is constant and equals to $\beta =\sum_{h = 1}^{3m}2x_h = 2mT$; 
    \item similarly, the total number of enemy relations between these groups is also always equaling $\gamma =\sum_{i=1}^{3m} 2x_h \cdot (X - x_h)$.
\end{enumerate}
It remains to determine the enemy relationships between element-cliques that are in the same coalition.
Assume w.l.o.g.\ that $\pi=(C_1, \dots, C_m )$ and $K^k_X\subseteq C_k$ and denote by $s_k = \sum_{h : \  K^h\subseteq C_k} x_h$. Hence, the enemy relationships between element-cliques are \\
$\sum_{k = 1}^m s_k(s_k - 1) - \sum_{h=1}^{3m} x_h(x_h - 1)$. Putting all together, $\forall \pi\in \Sigma$,
\begin{align*}
    \SW(\pi)= \alpha + \beta -\frac{ \gamma }{n} - \frac{1}{n}\left(\sum\limits_{k = 1}^m s_k(s_k - 1) - \sum\limits_{h=1}^{3m} x_h(x_h - 1) \right)\, .
\end{align*}

The only aspect affecting the $\SW$ of $\pi\in \Sigma$ is how $K^1, \dots, K^{3m}$ are located in its coalitions. In particular, the social welfare is maximum when $\sum_{k = 1}^m s_k(s_k - 1)$ is minimum.
Such a quantity is minimized for $s_1 = s_2 = \dots = s_m = T$.
In conclusion, the \threePartition\ instance is a ``yes'' instance if and only if in the social optimum of the reduced $\FA$ we have $s_1 = \dots = s_m = T$.
This proves \Cref{thm:OPThard}.

\section{An approximation mechanism}
For the sake of achieving NOM in polynomial time, in this section, we present a $(4+o(1))$-approximation mechanism. We recall that in~\cite{flammini2022strategyproof} it was shown that creating a coalition for each weakly connected component of $\Gf$ is SP and guarantees an $n$-approximation to the optimum. This is so far the best approximation achieved by an SP mechanism.  The bad performances of this mechanism can be attributed to the fact that when $\Gf$ is weakly connected but really sparse it would be convenient to split the unique weakly connected component of $\Gf$ into smaller coalitions, see for instance \Cref{ex::star}. To circumvent this problem, in~\cite{flammini2022strategyproof}, the authors presented a randomized mechanism, hereafter we will call it \randMech. It randomly splits the agents into two sets, each agent having probability $\frac{1}{2}$ of being in one of them, and then computes the weakly connected components on the two sides. \randMech\ is SP (in expectation) and guarantees an expected approximation $\leq 4$. 

Inspired by \randMech, we draw our deterministic and NOM mechanism. Specifically, it partitions the agents into two sets, $P_1$ and $P_2$, of size $\ceil{\frac{n}{2}}$ and $\floor{\frac{n}{2}}$, respectively. It then updates $P_1$ and $P_2$, through the subroutine \improveSW\ described by Algorithm~\ref{alg:approx} in the Appendix. \improveSW\ repeatedly tries to improve $\SW(\set{P_1, P_2})$ by swapping two agents, that is, simultaneously moving $i\in P_1$ to $P_2$ and $j\in P_2$ to $P_1$, or moving an agent from the largest to the smallest coalition (in case the two sets have the same size the algorithm will never perform move). \improveSW\  terminates when no swap or move can increase the $\SW$. The mechanism then computes the weakly connected components in $P_1$ and $P_2$ which will be the coalitions of the returned coalition structure. 

To show the mechanism is NOM, the initialization of $\set{P_1, P_2}$ will be crucial. Recall that $\delta(i)$ is the number of nodes weakly connected to $i$, while $N(X)$, for $X\subseteq \agents$, is the set of agents weakly connected to $X$. 
The mechanism will create the initial $\set{P_1, P_2}$ by greedily adding agents to the set $P_1$ in the following way: At first, it inserts an agent $i \in \agents$ of highest $\delta(i)$, then, iteratively proceeds by including an agent $j\in \neigh(P_1) \setminus P_1$ 
with highest $\delta(j)$ -- ties broken arbitrarily. This process proceeds until $P_1$ contains exactly $\lceil\frac{n}{2}\rceil$ agents. If at some point $\neigh(P_1) \setminus P_1= \emptyset$, the mechanism selects a new agent $i\in \agents\setminus P_1$ of highest $\delta(i)$, and proceeds as before. We call this partition a {\em greedy 2-partition} of $\agents$.
In summary:
\begin{mechanism}\label{mech:approx}
Given a set of agents $\agents$ and their declarations $\declared$, the mechanism creates a greedy 2-partition $\set{P_1, P_2}$. 
It then updates the partition using \improveSW:  $\set{P_1,P_2} \gets \improveSW(P_1,P_2)$.
Finally, it computes the weakly connected components in both $P_1$ and $P_2$, namely, $C_1, \dots C_m$, and returns $\pi = \set{C_1, \dots, C_m}$. 
\end{mechanism}

\begin{theorem}\label{thm:apxIsNOM}
 For $\FA$ instances, \Cref{mech:approx} is NOM and guarantees a $(4+o(1))$-approximation of the optimum in polynomial time.
\end{theorem}

Let $\pi^{\ref{mech:approx}} = \set{C_1, \dots, C_m}$ be the outcome of \ref{mech:approx}. We denote by $f_\pi$ the number of friendships within the coalitions in a partition $\pi$. 

\begin{observation}\label{obs:incresingFriendsIncresasWelfare}
    If $\pi$ and $\pi'$ are the partitions before and after the execution of a swap or move step by \improveSW\ during \ref{mech:approx}, then, $$\SW(\pi')-\SW(\pi) = (f_{\pi'} -f_\pi)\left(1 + 1/n\right)\, .$$
\end{observation}
\begin{proof}
    Let $s=\lceil\frac{n}{2}\rceil\cdot \left(\lceil\frac{n}{2}\rceil - 1\right) + \lfloor\frac{n}{2}\rfloor\cdot \left(\lfloor\frac{n}{2}\rfloor - 1\right)$ be the number of possible connections, both positive and negative, within the coalitions of $\pi$. 
    We notice that $s$ is also the number of possible connections within the coalitions of $\pi'$. In fact, a swap or move executed by the mechanism does not change the sizes and the number of coalitions.
    Therefore,
    $\SW(\pi)= f_\pi \left(1 +\frac{1}{n}\right) -\frac{s}{n}$, and the same holds for $\SW(\pi')$  replacing $f_\pi$ with $f_{\pi'}$. Hence, the thesis follows.
\end{proof}

In other words, a swap or a move in $\set{P_1, P_2}$ is convenient for the social welfare if and only if it strictly increases the number of positive relationships within coalitions. Such an observation implies that at most $f\leq n(n-1)$ swaps and moves will occur in any execution of \ref{mech:approx}, therefore, the mechanism is polynomial.

We next show the mechanism is NOM.

\begin{proof}[Proof \ref{mech:approx} is NOM]
    We show that for any agent $i$ there is no incentive to misreport her true preferences to improve either the best- or the worst-case scenario showing that both Condition 1 and 2 of the definition of NOM hold true.
 
    {\em Condition 1.} Recall, $F_i$ is the set of truthful friends of $i$ having size $f_i$. We start by noticing that, regardless of the declaration of $i$, in any outcome of the mechanism $i$ cannot be put in a coalition with more than $\min\set{f_i, \ceil{\frac{n}{2}} -1 }$, in fact, no coalition can have more than $\ceil{\frac{n}{2}}$ agents.
    We next show that if $i$ truthfully reports $t_i$, there exists $\declared_{-i}$ such that $i$ gets in $\ref{mech:approx}(t_i, \declared_{-i})$ a utility equal to  $\min\set{f_i, \ceil{\frac{n}{2}} -1 }$.
    
    If $f_i \leq \ceil{\frac{n}{2}} -1$, let $\declared_{-i}$ be such that $F_i$ is a bidirectional clique in $\Gf(t_i, \declared_{-i})$ and all agents in $F_i$ consider $i$ a friend. The remaining agents are isolated. In this case, it is easy to see, the mechanism will output $F_i\cup \set{i}$ in a coalition and the other agents in singletons. 
    
    If $f_i > \ceil{\frac{n}{2}} -1$, let $A \subseteq F_i$ such that $\modulus{A}= \ceil{\frac{n}{2}} -1$. We choose $\declared_{-i}$ so that $A$ is a bidirectional clique and $i$ is a friend for $j$ if and only if $j\in A$. No other friendship relationship exists. In this case, when initializing $P_1, P_2$, the mechanism will set $P_1= A \cup \set{i}$ and $P_2 = \agents\setminus P_1$. In fact, for any agent $j\not\in A \cup \set{i}$, we have $\delta(j)\leq 1$ while $ A \cup \set{i}$ is weakly connected and each agent $j' \in  A \cup \set{i}$ has $\delta(j')=\modulus{A}>1$. 
    The subroutine \improveSW\ will not change the partition $\set{P_1, P_2}$ as no agent in $A$ is weakly connected to $P_2$ and the number of positive between $i$ and $A$ is higher than the ones between $i$ and $P_2$, hence, no swap or move is profitable for the $\SW$.
    
    Putting all together, by truthfully reporting, it is possible for $i$ to end up in a coalition of value  $\min\set{f_i, \ceil{\frac{n}{2}} -1 }$. Since, regardless of the declaration of $i$, she cannot achieve a strictly higher utility in an outcome of \ref{mech:approx}, Condition 1 of NOM is satisfied.

    \noindent{\em Condition 2.} Recall $e_i$ is the size of truthful set of enemies $E_i$.  
    
    If $e_i \geq \ceil{\frac{n}{2}} -1$, given $A \subseteq E_i$ with $\modulus{A}= \ceil{\frac{n}{2}} -1$, for any possible declaration $d_i$ of $i$, we select $\declared_{-i}$ so that $A$ is a bidirectional clique and $i$ is a friend for $j$ if and only if $j\in A$. No other friendship relationship exists. Also in this case, the mechanism will set $P_1= A \cup \set{i}$ and $P_2 = \agents\setminus P_1$ and no swap and move will take place. Being $P_1$ weakly connected, $i$ will end up in a coalition with $ \ceil{\frac{n}{2}}-1 $ enemies. 
    Since the mechanism outputs coalitions of size at most $ \ceil{\frac{n}{2}}$, in no outcome of \ref{mech:approx} $i$ has a worse utility. Therefore, no misreport of $i$ can guarantee her a strictly better worst-case.

    If $e_i < \ceil{\frac{n}{2}} -1$, which implies, $e_i\leq \floor{\frac{n}{2}}-1$, in this case we start by showing that if $i$ truthfully reports $t_i$, then, $i$ cannot be in a coalition with less than $ \floor{\frac{n}{2}} - e_i - 1$ friends. In fact, let $\pi=\set{P_1,P_2}$ be the partition before \ref{mech:approx} computes the weakly connected components. Assume $i\in P_h$, for some $h\in [2]$. Then, the number of agents other than $i$ in $P_h$ is at least $ \floor{\frac{n}{2}} -1$. Thus, $\modulus{F_i \cap P_h}$ is minimized when $\modulus{P_h \setminus\set{F_i \cup \set{i}}}= e_i$, and hence $i$ has at least $ \floor{\frac{n}{2}} - e_i - 1$ friends in $P_h$. Since the coalition of $i$ will be the weakly connected component of $i$ in $P_h$, $i$ will be put in a coalition containing $F_i \cap P_h$. In conclusion, whatever the partition of $i$ is, by truthfully reporting, $i$ will always be in a coalition with at least $ \floor{\frac{n}{2}} - e_i - 1$ many friends.
    
    We next show that, regardless of the declaration of $i$, there exists $\declared_{-i}$ such that $i$ is put in a coalition with all her enemies and $ \floor{\frac{n}{2}} - e_i - 1$ many friends.
    Let $A= \agents\setminus  \set{E_i \cup X \cup\set{i}}$ where $X\subset F_i$ such that $\modulus{{E_i \cup X \cup\set{i}}}= \floor{\frac{n}{2}}$, which implies, $\modulus{A}= \ceil{\frac{n}{2}}$. 
    Let $\declared_{-i}$ be such that $A$ is a bidirectional clique, all agents in $E_i$ consider $i$ a friend. In this case, initially, the mechanism will set $P_1=A\setminus\set{j}\cup\set{i}$, for some $j\in A$. In fact, for this instance, $\delta(i)=n-1$, $\delta(j')=\modulus{A}>1$, for all $j'\in A$, and $\delta(j'')=1$, for all $j''\in \agents \setminus\set{A\cup\set{i}}$. Moreover, $A \cup\set{i}$ is weakly-connected, hence, computing the greedy 2-partition the mechanism selects at first $i$ and then $\ceil{\frac{n}{2}}-1$ agents in $A$. The mechanism will then improve the social welfare of $\set{P_1, P_2}$ with \improveSW. $\improveSW(P_1, P_2)$ will perform only the swap of $i$ and $j$ as this strictly increases the number of friendship relationships within coalitions -- no agent in $A$ considers $i$ a friend while all of them consider $j$ a friend, and $j$ is not connected to any agent in $A \cup\set{i}$ while $i$ possibly is. Once $P_1= A$, no other swap or move will occur -- a swap or move cannot strictly increase the number of positive relationships within coalitions.
    Therefore, the mechanisms will compute the weakly connected components in $P_1$ and $P_2$, being $P_2= E_i \cup X \cup\set{i}$ weakly connected $i$ will be put in the coalition $P_2$, and this is so, regardless of the declaration of $i$. This scenario is the worst possible for $i$ among the possible outcomes of \ref{mech:approx}, in fact, we have shown that $i$ cannot have strictly fewer friends and in this coalition the number of enemies is maximum. In conclusion, there is no way for $i$ to increase the worst outcome by misreporting her preferences.
    This shows that Condition 2 is satisfied and concludes our proof.
\end{proof}

To determine the approximation ratio of \ref{mech:approx}, we need to establish a lower bound for $f_{\pi^{\ref{mech:approx}}}$, the number of friendships within the coalitions of $\pi^{\ref{mech:approx}}$, w.r.t.\ the overall friendship relationships $f$.

\begin{lemma}\label{lemma:LBfriendsApprox}
For an $\FA$ instance with $f$ friendships, $f_{\pi^{\ref{mech:approx}}}\geq \frac{n-2}{2n-1}f$.
\end{lemma}
\begin{proof}
Let $\edg(S_1, S_2)$ be the number of edges between $S_1$ and $S_2$ and $\edg(S)$ be the number of edges within $S$, for $S_1, S_2, S\subseteq \agents$.

Let $\set{P_1, P_2}$ be the output of \improveSW\ during the execution of \ref{mech:approx}. 
When splitting $\set{P_1, P_2}$ into weakly connected components the number of friendship relationships within coalitions remains the same. Therefore, $f_{\pi^{\ref{mech:approx}}}$ equals $\edg(P_1)+\edg(P_2)$.
Moreover, when \improveSW\ terminates, a swap of two agents of $P_1$ and $P_2$ does not increase the $\SW$. Thus, for every $i\in P_1$ and $j\in P_2$,
\begin{align*}
    \edg(\{i\}, P_2 \setminus \{j\}) - \edg(\{i\}, P_1 \setminus \{i\}) + \\ \edg(\{j\}, P_1 \setminus \{i\}) - \edg(\{j\}, P_2 \setminus \{j\}) \leq 0 \, ,
\end{align*}
as, from \Cref{obs:incresingFriendsIncresasWelfare}, a swap is performed as long as it strictly increases the number of friendships within $P_1$ and $P_2$.

If we sum these inequalities for all $i \in P_1$ and $j \in P_2$ we obtain:
\begin{align*}
     (|P_2| - 1) \cdot \edg(P_1, P_2) - 2 |P_2| \cdot \edg(P_1) + \\ (|P_1| - 1)\cdot \edg(P_2, P_1) - 2|P_1|\cdot \edg(P_2) \leq 0 \, .
\end{align*}

Since $\edg(P_1, P_2) = f- f_{\pi^{\ref{mech:approx}}}$, $\modulus{P_1}= \ceil{\frac{n}{2}}$, and  $\modulus{P_2}= \floor{\frac{n}{2}}$,
\begin{gather*}
    (\lceil \frac{n}{2} \rceil + \lfloor \frac{n}{2} \rfloor - 2) \cdot (f- f_{\pi^{\ref{mech:approx}}})- 2 \lfloor \frac{n}{2} \rfloor \cdot \edg(P_1) - 2 \lceil \frac{n}{2} \rceil \cdot \edg(P_2) \leq 0 
\end{gather*}
and, using $\floor{\frac{n}{2}} \leq \ceil{\frac{n}{2}} \leq \frac{n+1}{2}$ and $\floor{\frac{n}{2}} + \ceil{\frac{n}{2}} = n$,  we finally derive
\begin{align*}
 (n - 2) (f- f_{\pi^{\ref{mech:approx}}})\leq 2 \ceil{\frac{n}{2}} \left( \edg(P_1) + \edg(P_2) \right) \leq (n+1)\cdot f_{\pi^{\ref{mech:approx}}} \, .
\end{align*}

In conclusion, $ \frac{2n-1}{n-2}  f_{\pi^{\ref{mech:approx}}}\geq f$. 
\end{proof}

This lemma constitutes the bulk of the proof of the approximation guarantee. We defer the details to the supplemental material. In a nutshell, we use essentially the same analysis as was made for \randMech, which has an expected approximation ratio of at most $4$. \randMech, however, guarantees that in expectation exactly $f/2$ positive edges are within coalitions of the outcome, while, by the above lemma, we only have at least $f\cdot\left(\frac{1}{2} -o(1)\right)$ positive relations are within coalitions leading to an approximation factor of $4 +o(1)$. In the supplemental material, we show there exists an instance where the approximation factor of \ref{mech:approx} is $4-o(1)$.

\section{Enemies Aversion preferences}
Enemies Aversion ($\EA$) preferences are the counterpart of $\FA$ where agents give priority to coalitions with fewer enemies, and when the number of enemies is the same, they prefer coalitions with a higher number of friends. This can be encoded in the class of ASHG with values $v_i(j)= \frac{1}{n}$, if $j\in F_i$, and $v_i(j)= -1$, otherwise. For this class, in~\cite{flammini2022strategyproof}, it has been shown that the optimum is hard to approximate within a factor $O(n^{1-\epsilon})$, for any positive $\epsilon$. Moreover, a poly-time and $O(n)$-approximating SP mechanism exists but, even if the time complexity is not a concern, strategyproofness and optimality are not compatible. It is therefore natural to wonder if optimality is compatible with NOM. Unfortunately, this is not the case.

\begin{theorem}\label{thm:EAnotNOM}
    For $\EA$ preferences, no optimum mechanism is NOM
\end{theorem}
\begin{proof}
Let us consider an instance where agent $i$ has exactly one enemy in $E_i$, the truthful set of enemies. If all the other agents declare everyone else as friends, the grand coalition is the social optimum and $i$ gets a utility of $\frac{n-2}{n} - 1 <0$. We do not know if this is the worst case, but, this means, there exists at least one outcome, when $i$ truthfully reports, in which $i$'s utility is strictly negative.

Assume now agent $i$ declares everyone is her enemy. Then, in any social optimum, $i$ ends up in the singleton coalition obtaining a utility $0$. So, this manipulation improves the worst case and violates Condition 2 in the definition of NOM, which makes the mechanism outputting the social optimum obviously manipulable.
\end{proof}

It has also been proven that an SP and $(1+\sqrt{2})$-approximating mechanism exists~\cite{flammini2022strategyproof}. It would be interesting to further investigate what are the boundaries of approximating SP or NOM mechanisms when time complexity is out of discussion.
\section{Conclusions}
In this paper, we investigated NOM in HGs with $\FA$ preferences, aiming at designing mechanisms optimizing the social welfare while preventing manipulation. Despite proving that computing a welfare-maximizing partition is NP-hard, we showed that a NOM mechanism computing the optimum always exits. In turn, for $\EA$ preferences, such a mechanism cannot exist.
To address the computational challenges of optimal outcomes under $\FA$, we presented a $(4+o(1))$-approximation mechanism that is NOM and runs in polynomial time. This mechanism not only improves on the best-known strategyproof mechanism, which provides a linear approximation in the number of agents, but also represents the first deterministic constant-factor approximation algorithm for $\FA$ preferences; this is an interesting contrast to $\EA$ preferences for which it is hard to approximate the optimum within a factor of $O(n^{1-\epsilon})$.

Interesting future research directions include the study of NOM for more general classes of HGs, for example, in ASHGs no bounded approximation is possible while requiring SP, therefore, it would be natural to consider a weaker notion of manipulability. Conversely, future work may focus on desirable properties like stability, welfare maximization (even beyond the utilitarian welfare), efficiency, or fairness, determining which kind of manipulations they are sensitive to.

\bibliographystyle{plain}
\bibliography{bibliography}
\appendix
\section{Appendix}
\subsection{Further Definitions}
In the viewpoint of graph representation, a collection of disjoint coalitions $C^1, \dots, C^m$ induce a cut in the graph $\Gf$, that is, $Cut(C^1, \dots, C^m)=\set{\set{i,j} \in F \,\vert\, i\in C^h, j\in C^\ell, h\neq \ell}$. The value of a cut is the value of friendship and enemy relationship contained in it, i.e., $\modulus{Cut(C^1, \dots, C^m)} - \frac{1}{n}\cdot \left(\sum_{h\neq \ell}\modulus{C^h}\cdot\modulus{C^l} - \modulus{Cut(C^1, \dots, C^m)}\right)$. We denote by $\SW(Cut(C^1, \dots, C^h))$ such value and observe that $\SW(C^1 \cup \dots \cup C^m)) - \sum_{h=1}^m \SW(C^h)$; in other words, the value of a cut $Cut(C^1, \dots, C^m)$ equals the loss or the gain we derive by splitting the agents in the coalitions $C^1, \dots, C^m$ rather than forming a unique coalition with all of them. 

We observe that if for a partition $\pi$, if there exist $C^1, \dots, C^h \in \pi$ such that $\SW(Cut(C^1, \dots, C^m))>0$, then, the outcome $\pi' $ obtained from $\pi$ by removing the coalitions ${C^1, \dots, C^m} $ and adding the coalition $ \cup_{h\in[m]}C^h$ is such that $\SW(\pi') > \SW(\pi)$. This motivates us to introduce the following property.

\begin{definition}[No Positive Cut]
    A partition $\pi$ satisfies the {\em no positive cut} (NPC) property if there exist no  $C^1, \dots, C^h \in \pi$ such that $\SW(Cut(C^1, \dots, C^m))>0$.
\end{definition}

Because of the previous discussion, we can conclude the following observation.

\begin{observation}\label{obs:OptNPC}
    Any optimal outcome satisfies the NPC property.
\end{observation}

\subsection{Preliminary Results on Optimal Outcomes}
\almostIsolatedOPT*

\begin{proof}
Let us assume by contradiction that there exists an optimal solution $\pi^*$  and $k$ coalitions $C_1, \cdots, C_k\in \pi^*$ such that $C \subseteq \cup_{h=1}^kC_h$ and $C_h \cap C\neq \emptyset$, for each $h\in [k]$.

We start by observing that if $i^*\notin C_h$, for some $h\in [k]$, then $C_h \cap C =C_h$. Otherwise, there are at least two weakly connected components in $C_h$, as $i^*$ is the only node connected to $C$ in $\Gf$, and splitting the weakly connected components always increases the social welfare.

Next, we prove $k \leq 2$. If not, there are at least two coalitions $C_{h}, C_{h'}$ not containing $i^*$. We previously showed such coalitions are made by only members of $C$, being $C$ a (mutual) clique the cut $(C_{h}, C_{h'})$ is positive, contradicting the optimality of $\pi^*$.

In the remaining, we show by contradiction that $k=2$ is not possible.  
Assume $k=2$, $i^*$ must belong to either $C_1$ or $C_2$, otherwise, we can apply the arguments of the previous paragraph and show that it is strictly profitable for the social welfare to merge these two coalitions.

Without loss of generality assume $i^*\in C_1$ and therefore we have $C_2\cap C =C_2$.
Let us denote by $c_1,c_2$ and $c$ the size of $C_1,C_2$ and $C$, respectively. In the cut $Cut(C_1,C_2)$ there are at least $2c_2(c - c_2)$ friendship relationships; all the remaining relationships are enemies. Therefore, the value of the cut is of at least 
\begin{align*}
\SW(Cut(C_1,C_2)) &\geq  2c_2(c -c_2) - \frac{2\cdot c_1\cdot c_2 -2c_2(c - c_2)}{n}\\
                  &= 2c_2(c - c_2) \left(1 +\frac{1}{n}\right) - \frac{2\cdot c_1\cdot c_2 }{n} \\
                  &=\frac{2c_2}{n}\cdot \left((c - c_2) \left(n+1\right) -c_1\right)>0 \; ,
\end{align*}
where the last inequality holds because $c_1 < n$.
Being $\pi^*$ optimal, the social welfare of the cut $Cut(C_1,C_2)$ cannot be positive, leading to a contradiction.

In conclusion, $k=1$ and the thesis follows.
\end{proof}

\octopusOPT*
\begin{proof}
Let us start by noticing that set $H$ is an almost isolated clique with hinge $i$. Therefore, by \Cref{lemma:almostIsolatedOPT}, in the social optimum all agents from $H$ will end up in the same coalition. Moreover, for each agent $k \in \agents\setminus\left(\{i\} \cup H\right)$, $N(k)\subseteq \set{i}$ --  $k$ can be weakly connected only to $i$ -- so if $k$ is not in the same coalition of $i$, then $k$ forms a singleton coalition. This leaves us with three possible types of optimal partition:
\begin{enumerate}
    \item[$\pi^1$] where $H$ and $i$ in the same coalition, while all remaining agents are in singletons;
    \item[$\pi^2$] where agents from $H$ form a coalition, while $i$ is in a different coalition together with some $C_1 \subseteq \agents\setminus\left(\{i\} \cup H\right)$, and remaining agents are in singletons;
    \item[$\pi^3$] where $H$, $i$ and some $C_2 \subseteq\agents\setminus\left(\{i\} \cup H\right)$ form one coalition, while all other agents are in singletons.
\end{enumerate} 

Denoted by $h=\modulus{H}$, $c_1=\modulus{C_1}$, $c_2=\modulus{C_2}$, and $g=\modulus{F_{i} \cap H}$ we next estimate the social welfare of $\pi^1$, $\pi^2$, and $\pi^3$.

In $\pi^1$, the only non-singleton coalition is $H\cup \set{i}$ in which there are $h-g$ enemy relationships, the ones of $i$ towards $H \setminus F_i$. Hence,
$    \SW(\pi^1)= h\left(h-1\right) + h + g - \frac{h - g}{n}
    =h\left(h-1\right) + h\left(1-\frac{1}{n}\right) + g\left(1+\frac{1}{n}\right)$.

In $\pi^2$, the coalition $H$ consists of only friends while in $C_1 \cup \set{i}$ there are at most $\modulus{C_1\cap F_i}\leq c_1$ friendship relationship, therefore,
\begin{align*}
    \SW(\pi^2)&\leq h\left(h-1\right) + c_1 - \frac{c_1}{n} - \frac{c_1\left(c_1-1\right)}{n}\\
              &= h\left(h-1\right) +  c_1\left(1-\frac{1}{n}\right)- \frac{c_1\left(c_1-1\right)}{n} \, .
\end{align*}

Finally, in $\pi^3$ again we have only one non-singleton coalition, $H\cup C_2\cup\set{i}$,  since $i$ has at most $c_2$ friendship relationship towards $C_2$, in this case we get
\begin{align*}
    \SW(\pi^3)& \leq h\left(h-1\right) + h + g - \frac{h - g}{n} + c_2 - \frac{c_2 + c_2\left(c_2-1\right) + 2h\cdot c_2}{n} \\
    &= \SW(\pi^1) +c_2 \left(1 - \frac{c_2 +2h}{n}\right) \, .
\end{align*}

We first show that $\SW(\pi^2) < \SW(\pi^1)$. Assuming the contrary, from the previous estimations we can deduce:
\begin{gather*}
h\left(1-\frac{1}{n}\right) + g \left(1+\frac{1}{n}\right) \leq c_1\left(1-\frac{1}{n}\right)- \frac{c_1\left(c_1-1\right)}{n} \, .
\end{gather*}

This implies, $h\left(1-\frac{1}{n}\right) \leq c_1\left(1-\frac{1}{n}\right)$ and therefore $h\leq c_1$. However, we know the head of the octopus contains at least $\ceil{\frac{n}{2}}$ nodes while $C_1 \subseteq \agents\setminus (\set{i}\cup H)$ implying $c_1 \leq n - h -1 \leq  n- \ceil{\frac{n}{2}} -1 < h$ -- a contradiction. 

Let us now show $\SW(\pi^3) < \SW(\pi^1)$. 

By contradiction, if $\SW(\pi^3) \geq \SW(\pi^1)$, from the previous estimations we deduce: $c_2 \left(1 - \frac{c_2 +2h}{n}\right) \geq 0$.
If $\pi^1\neq \pi^3$, then, $c_2\geq 1$. Therefore, $c_2 +2h \leq n$ must hold true. Being, $h \geq \ceil{\frac{n}{2}}\geq \frac{n}{2}$, we can conclude $c_2 +2h \geq n +1$ leading again to a contradiction.

To sum up, if $h \geq \ceil{\frac{n}{2}}$, there exists a unique social optimum consisting of the coalition $\set{i}\cup H$ and all remaining agents partitioned in singleton coalitions.
\end{proof}

\begin{lemma}\label{lemma:genOctopusOPTtentaclesForMech}
Let $\Gf$ be a generalized $i$-centered octopus graph with head $H$ and tentacles $T_1, \dots, T_m$. If $\pi^*=\set{H\cup\set{i}, T_1, \dots, T_m}$ is an optimal outcome of \ref{mech:OPT}, then, for each $l\in[m]$, $|H| > \frac{\modulus{F_{i} \cap  T_l}}{|T_l|}\cdot \frac{n+1}{2} - 1$.
\end{lemma}
\begin{proof}
Let us set $g_l=\modulus{F_{i} \cap  T_l}$.
Consider the partition $\pi$ obtained from $\pi^*$ by merging the coalitions $H\cup\set{i}$ and $T_l$. Since $\pi^*$ is optimal,
$
 \SW(\pi) - \SW(\pi^*) = g_l- \frac{2|T_l| - g_l}{n} - \frac{2|T_l|\cdot|H|}{n} <  0 \, .
$
The inequality is strict, because otherwise \ref{mech:OPT} choosing among the partitions with the same social welfare the ones with the smallest number of coalitions would merge $H\cup\set{i}$ and $T_l$.
Therefore, $\forall l \in [m]$,
$
|H| > \frac{g_l(n+1)}{2|T_l|} - 1$, 
and hence it holds for the maximum.
\end{proof}

\subsection{Computing the optimum is NP-hard: Proof of \Cref{thm:mainThmReduction}}

\Cref{thm:mainThmReduction} is the result of the lemmas proven in this section.

\begin{lemma}\label{lemma:no2K_XsameCoalition}
If $\pi^*$ is a socially optimum partition, then, $K^j_X \cup K^{j'}_X\not \subseteq C$, for any $C\in  \pi^*$ and $j \neq j'$.
\end{lemma}

\begin{proof}
Let us suppose the contrary. Assume there exists $C\in \pi^*$ and $j,j'$, with $j\neq j'$, such that $K^j_X \cup K^{j'}_X \subseteq C$.
Let us now consider the partition $\pi= \pi^* \setminus\set{C} \cup \set{C\setminus K^j_X\cup K^j_X}$. 

Since $\pi^*$ is socially optimal, we have $\Delta = \SW(\pi) - \SW(\pi^*) \leq 0$. 
In turn, $\Delta$ can be estimated using the following observations:
\begin{itemize}
    \item the total number of friendship relations between $K_X^j$ and the rest of the graph is exactly  $\sum_{i=1}^{3m} 2x_i = 2mT$, implying that, by separating $K_X^j$ from $C\setminus K^j_X$, it is not possible to lose more than $2mT$ positive relations;
    \item all members of $K^j_X$ and $K^{j'}_X$ belong to $C$, therefore, by splitting $C$ into $K_X^j$ and $C\setminus K^j_X$ we will remove at least $2|K^j_X|\cdot|K^{j'}_X| =2X^2$ enemy relationships.
\end{itemize}

Putting these two considerations together and using $X = 4m^2T$, we get: 
\begin{align*}
\Delta \geq \frac{2X^2}{n} - 2mT = \frac{32m^4T^2}{4m^3T+mT} - 2mT = \frac{2mT(8m^2 - 1)}{4m^2+1} \, ,
\end{align*}
which is strictly positive for any $m \geq 1$. This leads to contradiction to the optimality of $\pi^*$ and the thesis follows.
\end{proof}

\begin{lemma}\label{lemma:K_XoneCoalitio}
Let $\pi^*$ be a socially optimum partition. For each $j\in[m]$, there exists $C\in\pi^*$ such that $K^j_X \subseteq C$.
\end{lemma}
\begin{proof}
By contradiction, assume $K^j_X \not\subseteq C$ for some $j$, then, there exists a partition of $K^j_X$, $A_1, \dots, A_k$, with $k\geq 2$, and $k$ distinct coalitions in $\pi^*$, $C_1, \dots, C_k$, such that $A_h \subseteq C_h$, for all $h\in[k]$.

Let us compare $\pi^*$ with the partition $\pi$ derived from $\pi^*$ by replacing every $C_h$ by $C_h\setminus A_h$ and adding a separate coalition $K^j_X$. Then, the following hold:
\begin{itemize}
    \item by grouping $A_1, \dots , A_k$ into one coalition we will gain at least $2(X-1)$ friendship relations\footnote{Given an undirected graph  and a clique of size $x$, for any non-trivial partition of the clique there are at least $x-1$ edges in the cut.};
    \item the total number of positive connections between $K^i_X = \bigcup_{h\in[k]} A_h$ and the remaining graph is $2mT$ implying that by separating $A_1, \dots , A_k$ from $C_1\setminus A_1, \dots, C_k\setminus A_k$, respectively, we will lose at most $2mT$ friendship relationship in total. 
\end{itemize}

Consequently, the difference in welfare is given by
\begin{align*}
    \SW(\pi) - \SW(\pi^*) \geq 2(X-1) - 2mT  &= 8m^2T - 2 - 2mT \\
                                             &= 2mT(4m - 1) - 2 \, ,
\end{align*}
which is strictly positive for any $m, T \geq 1$. A contradiction to the optimality of $\pi^*$, therefore, the thesis follows.
\end{proof}

\begin{lemma}\label{lemma:KiAllTogether}
Let $\pi^*$ be a socially optimum partition. For each $K^i$, there exists $C\in\pi^*$ such that $K^i \subseteq C$.
\end{lemma}
\begin{proof}
The statement holds true if $x_i=1$ as the clique consists of a unique member. Hereafter, we assume $x_i>1$.

By \Cref{lemma:no2K_XsameCoalition,lemma:K_XoneCoalitio}, if $C\in\pi^*$, then $C$ contains at most one $K^j_X$, for some $j$. Hence, if there exists $v\in K^i\cap C$, then, $v$ has at most one friendship relationship towards an agent in $C\setminus K^i$, i.e., the friendship relationship connecting $v$ to the $K_X^j$ that is possibly contained in $C$ (it is possible that $C$ does not contain any clique of size $X$). 
 
Assume that $K^i \not\subseteq C$, then, there exists a partition of $K^i$, $A_1, \dots, A_k$, with $k\geq 2$, and $k$ distinct coalitions in $\pi^*$, $C_1, \dots, C_k$, such that $A_h \subseteq C_h$, for all $h\in[k]$.

If one of these coalitions, say $C_{\ell}$, contains one of the cliques of size $X$, say $K^j_X$, let us create the new partition $\pi$ by removing $A_h$ from $C_h\setminus A_h$, for each $h \neq \ell$, and extending the coalition $C_{\ell}$ to the coalition $C'_{\ell} = C_{\ell} \cup (\bigcup_{h\neq \ell} A_h) = C_\ell \cup K^i$. In this case, the difference in social welfare between $\pi^*$ and $\pi$ can be estimated as follows:
First notice that by putting $A_1, \dots, A_k$ into one coalition we will gain at least $2(x_i-1)$ friendship relations since this number is minimized when $k=2$ and in this case equals to $2|A_1| \cdot |A_2| = 2|A_1| \cdot \left(x_i - |A_1|\right) \geq 2(x_i-1)$. Moreover, by transforming $\pi^*$ into $\pi$ we will add a number of enemy relationships that is at most $ 2|C'_{\ell}\setminus K^i|\cdot (x_i-1) $ as at least one member of $K^i$ was already in $C_{\ell}$. Since in $C'_{\ell}$ there is only one clique of size $X$ and possibly all $K^{i'}$ for $i'\in [3m]$ we can conclude
\begin{align*}
    |C'_{\ell}\setminus K^i| < X + \sum_{i'\neq i} |K^{i'}| = X + \sum_{i'\neq i} x_{i'} < X + mT
\end{align*}
 and hence the number of enemy relations added by transforming $\pi^*$ into $\pi$ is less than $2(x_i - 1)\cdot(X + mT)$.

Altogether, we obtain the following estimation of $\Delta = \SW(\pi) - \SW(\pi^*)$:
\begin{align*}
    \Delta &> 2 (x_i - 1) - \frac{2(x_i - 1) \cdot (X + mT)}{n} = \\
    &= 2(x_i - 1) \cdot \frac{n - X - mT}{n} = \\
    &= 2(x_i - 1) \cdot \frac{4m^2T(m - 1)}{n} \, ,
\end{align*}
which is strictly positive when $x_i > 1$ and $m > 1$ -- a contradiction to the optimality of  $\pi^*$.

If, instead, there is no $C_{\ell}$ containing a clique of size $X$, let us derive $\pi$ from $\pi^*$ by simply separating $A_h$ from $C_h\setminus A_h$, for every $h$, and creating a new coalition $C = \bigcup_{h\in [k]} A_h = K^i$. Then, since all positive connections between $K^i$ and the rest of the graph lead to the cliques of size $X$, we do not lose any friendship relations. Conversely, by uniting $A_1, \dots , A_k$ we gain at least $2(x_i - 1)$ positive edges. Thus, $\Delta \geq 2(x_i - 1)>0$, and $\pi^*$ again cannot be the social optimum.
\end{proof}

With the next lemma, we show that every coalition in the social optimum must contain a clique of size $X$.

\begin{lemma}
    If $\pi^*$ is a socially optimal partition, then for every $C\in \pi^*$ there exists $j\in[m]$ such that $K^j_X \subseteq C$.
\end{lemma}
\begin{proof}
Let us suppose in the optimal partition $\pi^*$ there exists a coalition, say $C_1$, such that it does not contain a clique of size $X$ and, therefore, by \Cref{lemma:KiAllTogether}, $C_1 = \bigcup_{i \in S} K^i$, for some $S \subseteq \{1, \dots, 3m\}$. Then, let us take an arbitrary coalition $C_2$ containing a clique $K^j_X$, for some $j$, and consider the partition $\pi$ obtained from $\pi^*$ by merging $C_1$ and $C_2$ into one coalition $C' = C_1 \cup C_2$. The comparison between $\SW(\pi)$ and $\SW(\pi^*)$ can be obtained thanks to following observations:
\begin{itemize}
    \item between $C_1$ and $K^j_X$ there are exactly $\sum_{i \in S} 2x_i$ friendship relations, therefore, we will gain this number of positive edges in $\pi$;
    \item in $C_2$ there is exactly one clique of size $X$ and at most all $K^i$ for $i \notin S$, implying that, by merging $C_1$ and $C_2$, we will add to the partition at most  $2X\cdot \sum_{i \in S}  x_i - \sum_{i \in S} 2x_i + 2\sum_{i \in S} x_i \sum_{k \notin S} x_k = 2(X - 1)\cdot \sum_{i \in S} x_i +  2\sum_{i \in S} x_i \sum_{k \notin S} x_k$ negative edges.
\end{itemize}
Overall, with $s$ denoting  $\sum_{i \in S} x_i$:

\begin{align*}
\Delta =  \SW(\pi) - \SW(\pi^*) &\geq 2s - \frac{2(X - 1)\cdot s + 2s(mT - s)}{n}\\
& \geq \frac{2s}{n}\left( n - X + mT-s\right)\\
& = \frac{2s}{n}\left( 4m^3T + mT - 4m^2T+ mT-s\right) \\
& \geq \frac{2s}{n}\cdot mT > 0
\end{align*}
where the second last inequality holds true because $4m^3T \geq 4m^2T$ and $mT-s\geq 0$.

But for positive $s$, $m$, and $T$ this expression is always positive, which means that $\pi^*$ is not socially optimal -- a contradiction.
\end{proof}


\subsection{Approximate Mechanism}

Here we formally report the pseudocode of \improveSW, see Algorithm~\ref{alg:approx}.

\begin{algorithm}[h!]\DontPrintSemicolon
\caption{$\improveSW(X, Y)$}\label{alg:approx}
\Input{A partition $\set{X,Y}$ of $\agents$ s.t. $\modulus{X}= \ceil{\frac{n}{2}}$, $\modulus{Y}= \floor{\frac{n}{2}}$}
\tcc{Swap pairs of agents to increase $\SW$}
\Repeat{
 can no longer {\bf Swap} any pair $i\in X$, $j\in Y$
}{
\If{$\exists i\in X, j\in Y $ s.t. $\SW(\set{ X \setminus\set{i}\cup\set{j}, Y\setminus\set{j}\cup\set{i}}) >\SW(\set{X,Y})$}{
$X \gets X \setminus\set{i}\cup\set{j}$, $Y\gets Y\setminus\set{j}\cup\set{i}$ \tcp*{{\bf Swap}}
}
}
\tcc{Move an agent from largest to smallest coalition to increase $\SW$}
\lIf{$\modulus{X}=\modulus{Y}$}{\Return{$\set{X,Y}$} \tcp*{\improveSW\ }}
\Repeat{
Can no longer {\bf Move} an agent from smallest to largest coalition
}{
\If{$\exists i\in X \vert \SW(\set{X\setminus\set{i},Y\cup\set{i}}) >\SW(\set{X,Y})$}{
$Y\gets X\setminus \set{i}$, $X \gets Y \cup \set{i}$ \tcp*{{\bf Move}}
}
}

\Return{$\set{X,Y}$}
\end{algorithm}

\begin{theorem}
    Mechanism~\ref{mech:approx} guarantees an $4 +o(1)$-approximation of $\opt$.
\end{theorem}
\begin{proof}
    Let us start by noticing that a trivial upper-bound of $\opt$ is $f$, that is, the number of friendship relationships for the given $\FA$ instance.
    
    Let $\pi = \set{C_1, \dots C_m}$ be the outcome of the mechanism, and let us denote by $f_\pi$ the number of positive relationships within the coalitions of $\pi$. Furthermore, we denote by $c_h$ the size of $C_h$. By \Cref{prop-utility-coaltion} we have $\SW(\pi)= f_\pi\left(1 + \frac{1}{n}\right) - \sum_{h=1}^m \frac{c_h(c_h-1)}{n}$.
    
    Notice that $\sum_{h=1}^m c_h = \modulus{P_1} +\modulus{P_2}$.
    The quantity $\sum_{h=1}^m \frac{c_h(c_h-1)}{n}$ is maximized when $P_1$ and $P_2$ are both weakly connected, if possible.  Therefore, we next give lower bounds to $\SW(\pi)$ depending on the value of $f_{\pi}$.
    
    Case 1: if $f_\pi \leq \ceil{\frac{n}{2}} -1$, the worst case scenario occurs when there is a unique weakly connected component in $P_1$, which is also a tree, in that case, 
    $$\SW(\pi)\geq f_\pi\left(1 + \frac{1}{n}\right) - \frac{f_\pi(f_\pi+1)}{n}= f_\pi \left( 1- \frac{f_\pi}{n}\right) \, ;$$
    
    Case 2: if $\ceil{\frac{n}{2}}\leq f_\pi \leq n-2$, then, 
    \begin{align*}
        f_\pi\left(1 + \frac{1}{n}\right) -  \frac{\ceil{\frac{n}{2}}\left(\ceil{\frac{n}{2}}-1\right)}{n} - \frac{(f_\pi - \ceil{\frac{n}{2}} +2 )(f_\pi - \ceil{\frac{n}{2}} +1)}{n} \, .
    \end{align*}
    
    Case 3: if $f_\pi > n-2$ we have
    $\SW(\pi)\geq f_\pi\left(1 + \frac{1}{n}\right) - \frac{\floor{\frac{n}{2}}\left(\floor{\frac{n}{2}}-1\right)}{n} - \frac{\ceil{\frac{n}{2}}\left(\ceil{\frac{n}{2}}-1\right)}{n} = f_\pi\left(1 + \frac{1}{n}\right) - \frac{\floor{\frac{n}{2}}^2 +\ceil{\frac{n}{2}}^2}{n} +1 .$

Let us denote by $R$ the approximation ratio of the mechanism. Employing the bound on $f_\pi$ obtained in \Cref{lemma:LBfriendsApprox}.

In case 1, where $f_\pi \leq \ceil{\frac{n}{2}} -1 \leq \frac{n-1}{2}$, we obtain 
\begin{align*}
    R&\leq \frac{ \frac{2n-1}{n-2} f_\pi}{ f_\pi \left( 1- \frac{f_\pi}{n}\right)} 
    \leq \frac{ n(2n-1)}{ (n-2)(n- f_\pi)}  \leq \frac{ 2n(2n-1)}{ (n-2)(n-1)} = 4 +o(1) \, . 
\end{align*}

In case 2, 
\begin{gather*}
     R \leq \frac{ \frac{2n-1}{n-2} f_\pi}{f_\pi\left(1 + \frac{1}{n}\right) -  \frac{\ceil{\frac{n}{2}}\left(\ceil{\frac{n}{2}}-1\right)}{n} - \frac{(f_\pi - \ceil{\frac{n}{2}} +2 )(f_\pi - \ceil{\frac{n}{2}} +1)}{n}}  = \\ \frac{n\cdot(2n-1)}{(n-2)\left(n+1 - \frac{\ceil{\frac{n}{2}}\left(\ceil{\frac{n}{2}}-1\right)}{f_\pi} - \frac{(f_\pi - \ceil{\frac{n}{2}} +2 )(f_\pi - \ceil{\frac{n}{2}} +1)}{f_\pi}\right)} = \\ \frac{n\cdot(2n-1)}{(n-2)\left(n+1  - \frac{2\left(\ceil{\frac{n}{2}}\right)^2 - 4\ceil{\frac{n}{2}} + f^2_\pi + 2 + 3f_\pi -2f_\pi\ceil{\frac{n}{2}}}{f_\pi}\right)} = \\ \frac{n\cdot(2n-1)}{(n-2)\left(n+1  - f_\pi - 3 + 2\ceil{\frac{n}{2}} - \frac{2\left(\ceil{\frac{n}{2}} -1\right)^2}{f_\pi}\right)}
\end{gather*}

To estimate this expression from above, we should minimize $-f_\pi - \frac{2\left(\ceil{\frac{n}{2}} -1\right)^2}{f_\pi}$. Its derivative with respect to $f_\pi$ equals to:
\begin{align*}
    -1 + \frac{2\left(\ceil{\frac{n}{2}} -1\right)^2}{f^2_\pi}
\end{align*}
and therefore equals to zero when $f_\pi = \pm \sqrt{2}\left(\ceil{\frac{n}{2}} -1\right)$. So, this expression increases between $\ceil{\frac{n}{2}}$ and $\sqrt{2}\left(\ceil{\frac{n}{2}} -1\right)$ and decreases between $\sqrt{2}\left(\ceil{\frac{n}{2}} -1\right)$ and $n-2$, and reaches its minimum either in $\ceil{\frac{n}{2}}$ or in $n-2$.

1. $f_\pi = \ceil{\frac{n}{2}}$:
\begin{gather*}
    \frac{n\cdot(2n-1)}{(n-2)\left(n+1  - \ceil{\frac{n}{2}} - 3 + 2\ceil{\frac{n}{2}} - \frac{2\left(\ceil{\frac{n}{2}} -1\right)^2}{\ceil{\frac{n}{2}}}\right)} = \\ \frac{n\cdot(2n-1)}{(n-2)\left(n+ 2 - \ceil{\frac{n}{2}} - \frac{2}{\ceil{\frac{n}{2}}}\right)} \leq \frac{n\cdot(2n-1)}{(n-2)\left(n+ 2 - \frac{n-1}{2} - \frac{4}{n-1}\right)} = \\ \frac{2n\cdot(2n-1)(n-1)}{(n-2)(n^2 +4n - 13)} = 4 - o(1)
\end{gather*}

2. $f_\pi = n-2$:
\begin{gather*}
    \frac{n\cdot(2n-1)}{(n-2)\left(n+1  - n + 2 - 3 + 2\ceil{\frac{n}{2}} - \frac{2\left(\ceil{\frac{n}{2}} -1\right)^2}{n-2}\right)} = \\ \frac{n\cdot(2n-1)(n-2)}{(n-2)\left(2\ceil{\frac{n}{2}}(n-2) - 2\left(\ceil{\frac{n}{2}} -1\right)^2\right)} = \\ \frac{n\cdot(2n-1)}{2n\ceil{\frac{n}{2}} - 4\ceil{\frac{n}{2}} - 2\left(\ceil{\frac{n}{2}}\right)^2 + 4\ceil{\frac{n}{2}} - 2} = \\ \frac{n\cdot(2n-1)}{2\ceil{\frac{n}{2}}\cdot\left(n - \ceil{\frac{n}{2}}\right) - 2} \leq \frac{n\cdot(2n-1)}{\frac{n^2-1}{2} - 2} = \frac{2n\cdot(2n-1)}{n^2 - 5} = 4 + o(1)
\end{gather*}

  Finally, in case 3, 
  \begin{align*}
    R&\leq \frac{ \frac{2n-1}{n-2} f_\pi}{f_\pi\left(1 + \frac{1}{n}\right) - \frac{\floor{\frac{n}{2}}^2 +\ceil{\frac{n}{2}}^2}{n} +1}  
\end{align*}
   this ratio is decreasing in $f_\pi$, therefore 
    \begin{align*}
    R&\leq \frac{ n\frac{2n-1}{n-2} (n-1)}{(n-1)\left(n + 1\right) - \floor{\frac{n}{2}}^2 - \ceil{\frac{n}{2}}^2 +n}  \leq \\
    &\leq \frac{ 2n\frac{2n-1}{n-2} (n-1)}{n^2 -3 +2n}
\end{align*}
   it is easy to see that in this case $R\leq 4- o(1)$.
\end{proof}

\section{Lower bound for \ref{mech:approx}}
In this section,  we provide a $\FA$ instance where the approximation of \ref{mech:approx} is at least $4- o(1)$.

Let us construct a friendship relations graph for $n$ agents where for every $1 \leq i \leq \ceil{\frac{n}{2}} - 1$ there are edges from $i$ to $i + 1$ and from $i$ to $\ceil{\frac{n}{2}} + i$ and no other edges.

Since the first $\ceil{\frac{n}{2}} - 1$ nodes have a degree $2$, while others have a degree $1$, on the initialization step the mechanism will put them into $P_1$ together with one of the remaining nodes. This partition is stable and therefore in the output there will be a coalition of size $\ceil{\frac{n}{2}}$ corresponding to $P_1$ (since it is weakly connected) and $\floor{\frac{n}{2}}$ singleton coalitions into which $P_2$ will fall apart.

In the coalition of size $\ceil{\frac{n}{2}}$ there are exactly $\ceil{\frac{n}{2}} - 1$ edges, so, the social welfare of the mechanism's output is:
\begin{align*}
    \left(\ceil{\frac{n}{2}} - 1\right)\cdot\left(1 + \frac{1}{n}\right) - \frac{\ceil{\frac{n}{2}}\cdot\left(\ceil{\frac{n}{2}} - 1\right)}{n} = \left(\ceil{\frac{n}{2}} - 1\right)\cdot \frac{\floor{\frac{n}{2}} + 1}{n} \leq \\ \frac{n^2 - 1}{4n}
\end{align*}

Let us note that for any set of even positive numbers $k_1, \dots, k_{m-1}$ and a positive number $k_m$ such that $\sum_{i=1}^m k_i = n$ this graph can be partitioned into $m$ coalitions $C_1, \dots, C_m$ with $|C_i| = k_i$ and $C_i$ containing exactly $k_i - 1$ edges. If for every $i$ we put into $C_i$ the vertices with the numbers from the interval $\left[\sum_{j=1}^{i-1} \frac{k_j}{2} + 1; \sum_{j=1}^{i-1} \frac{k_j}{2} + \frac{k_i}{2}\right]$ and from the interval $\left[\ceil{\frac{n}{2}} + \sum_{j=1}^{i-1} \frac{k_j}{2} + 1; \ceil{\frac{n}{2}} + \sum_{j=1}^{i-1} \frac{k_j}{2} + \frac{k_i}{2}\right]$, the resulting partition will satisfy these properties. One might easily check that this is a proper partition, i.e. all the vertices are covered and for all $i \neq j$ it holds that $C_i \cap C_j = \emptyset$, and from the way the graph was constructed it directly follows that $C_i$ is weakly connected and contains $|C_i|$ edges. Moreover, $C_i$ consists of exactly $2\frac{k_i}{2} = k_i$ vertices.

Therefore, there exists a partition in which there are $\floor{\frac{n}{2\floor{\frac{\sqrt{n+1}}{2}}}}$ coalitions of size $2\floor{\frac{\sqrt{n+1}}{2}}$ (note that this number is even) and one coalition of size $n - \floor{\frac{n}{2\floor{\frac{\sqrt{n+1}}{2}}}}\cdot 2\floor{\frac{\sqrt{n+1}}{2}}$.

Its social welfare is at least equal to:
\begin{align*}
    \floor{\frac{n}{2\floor{\frac{\sqrt{n+1}}{2}}}} \cdot \left(2\floor{\frac{\sqrt{n+1}}{2}} - 1\right)\cdot \frac{n - 2\floor{\frac{\sqrt{n+1}}{2}} + 1}{n} > \\ \frac{n}{2\floor{\frac{\sqrt{n+1}}{2}}} \cdot (\sqrt{n+1} - 1) \cdot \frac{n - \sqrt{n+1} - 1}{n} > \\ \frac{\sqrt{n+1} - 1}{\sqrt{n+1} + 1} \cdot (n - \sqrt{n+1} - 1)
\end{align*}

Thus, the ratio between $\opt$ and the social welfare of the mechanism's outcome is greater or equal to:
\begin{align*}
    \frac{\frac{\sqrt{n+1} - 1}{\sqrt{n+1} + 1} \cdot (n - \sqrt{n+1} - 1)}{\frac{n^2 - 1}{4n}} = \frac{4n\cdot (\sqrt{n+1} - 1) (n - \sqrt{n+1} - 1)}{(n^2-1)(\sqrt{n+1} + 1)} =\\ 4 - O\left(\frac{1}{\sqrt{n+1}}\right)
\end{align*}

\end{document}